\documentclass[lettersize,journal]{IEEEtran}
\usepackage{amsmath,amsfonts}
\usepackage{algorithm}
\usepackage{algpseudocode}
\usepackage{array}
\usepackage{subcaption}
\usepackage{textcomp}
\usepackage{stfloats}
\usepackage{url}
\usepackage{verbatim}
\usepackage{graphicx}
\usepackage{cite}
\usepackage{soul}
\usepackage{xcolor}
\usepackage{tcolorbox}
\usepackage{afterpage}
\usepackage{makecell}
\usepackage{threeparttable}
\usepackage{bm}
\usepackage{dsfont}
\usepackage{xcolor}
\usepackage{graphicx}
\usepackage{amssymb}
\usepackage{booktabs}
\usepackage{multirow}   
\usepackage{array}      
\usepackage{enumitem}

\usepackage{mathrsfs}  

\newtheorem{definition}{Definition}

\newtheorem{theorem}{Theorem}

\newtcolorbox{highlighted}{colback=yellow!50,colframe=yellow!50,arc=0pt,outer arc=0pt}

\begin{document}

\title{Digital Twin-Assisted In-Network and Edge Collaboration for Joint User Association, Task Offloading, and Resource Allocation in the Metaverse}

\author{Ibrahim Aliyu,~\IEEEmembership{Member,~IEEE}, Seungmin Oh , Sangwon Oh  and Jinsul Kim~\IEEEmembership{Member,~IEEE}
\thanks{This work was supported in part by the Institute of Information Communications Technology Planning Evaluation (IITP) grant funded by the Korean government (MSIT; RS-2025-00345030, Development of digital twin base network failure prevention and operation management automation technology, 50\%) and in part by the Innovative Human Resource Development for Local Intellectualization program through the IITP grant funded by the Korea government (MSIT; IITP-2025-RS-2022-00156287, 50\%).\textit{Corresponding authors: Jinsul Kim (email: jsworld@jnu.ac.kr)}.}
\thanks{Ibrahim Aliyu, Seungmin Oh, Sangwon Oh  and Jinsul Kim are with the Department of Intelligent Electronics and Computer Engineering, Chonnam National University, Gwangju 61186, Korea.}
}

\maketitle

\begin{abstract}
Advancements in extended reality (XR) are fueling the Advancements in extended reality (XR) are driving the development of the metaverse, which demands efficient real-time transformation of 2D scenes into 3D objects, a computation-intensive process that necessitates task offloading because of complex perception, visual, and audio processing. This challenge is further compounded by asymmetric uplink (UL) and downlink (DL) data characteristics, where 2D data are transmitted in the UL and 3D content is rendered in the DL. 
To address this issue, we propose a digital twin (DT)-based in-network computing (INC)-assisted multi-access edge computing (MEC) framework that enables real-time synchronization and collaborative computing via URLLC. In this framework, a network operator manages wireless and computational resources for XR user devices (XUDs), while XUDs autonomously offload tasks to maximize their utilities. We model the interactions between XUDs and the operator as a Stackelberg Markov game, where the optimal offloading strategy constitutes an exact potential game with a Nash Equilibrium (NE), and the operator’s problem is formulated as an asynchronous Markov decision process (MDP). 
We further propose a decentralized solution in which XUDs determine offloading decisions based on the operator’s joint UL-DL optimization of offloading mode (INC-E or MEC only) and DL power allocation. A Nash-asynchronous hybrid multi-agent reinforcement learning (AMRL) algorithm is developed to predict the UL user-associated and DL transmission power, thereby achieving NE. Simulation results demonstrate that the proposed approach considerably improves system utility, uplink rate, and energy efficiency by reducing latency and optimizing resource utilization in the metaverse environments.

\end{abstract}

\begin{IEEEkeywords}
Computation offloading, digital twin, deep reinforcement learning, game theory, in-network computing, multi-access edge computing
\end{IEEEkeywords}

\maketitle

\section{Introduction}

The rapid advancement of extended reality (XR) technologies within the metaverse is transforming gaming, workplace collaboration, and scientific research \cite{rosenberg2022}. XR applications require the real-time offloading of computationally intensive tasks from XR user devices (XUDs) to distributed computing resources to support immersive 3D rendering. Meanwhile, integrating digital twins (DTs) into edge computing creates new opportunities to enhance resource allocation through intelligence, efficiency, and flexibility  \cite{aliyu2024towards}. Although multi-access edge computing (MEC) has emerged as a prominent offloading paradigm, it suffers from resource contention, high signaling overhead, and limited scalability  \cite{aliyu2023toward, chen2022dynamic, pham2022partial}. Collaborative INC has recently been proposed as a promising complement to MEC; however, it introduces critical challenges. In particular, asymmetric uplink (UL) and downlink (DL) data sizes, where UL transmits 2D data, and DL delivers reconstructed 3D content, lead to increased inefficiencies and latency  \cite{dvorovzvnak2020monster}. Moreover, dynamic XUD task arrivals, interference-prone wireless channels, and energy-intensive DL transmissions further aggravate offloading inefficiency and user association imbalance  \cite{tutuncuoglu2024sample, yan2021pricing}. These challenges motivate the development of a unified framework that incorporates interference-aware channel selection, adaptive offloading modes, task splitting, and energy-aware downlink control to enable scalable, low-latency XR services.  Fig.~\ref{fig:INC-E_arch} illustrates an overview of the proposed DT-assisted INC-enabled architecture and contrasts it with conventional MEC-based offloading.

\begin{figure}[t]
  \centering
  \includegraphics[width=0.4\textwidth]{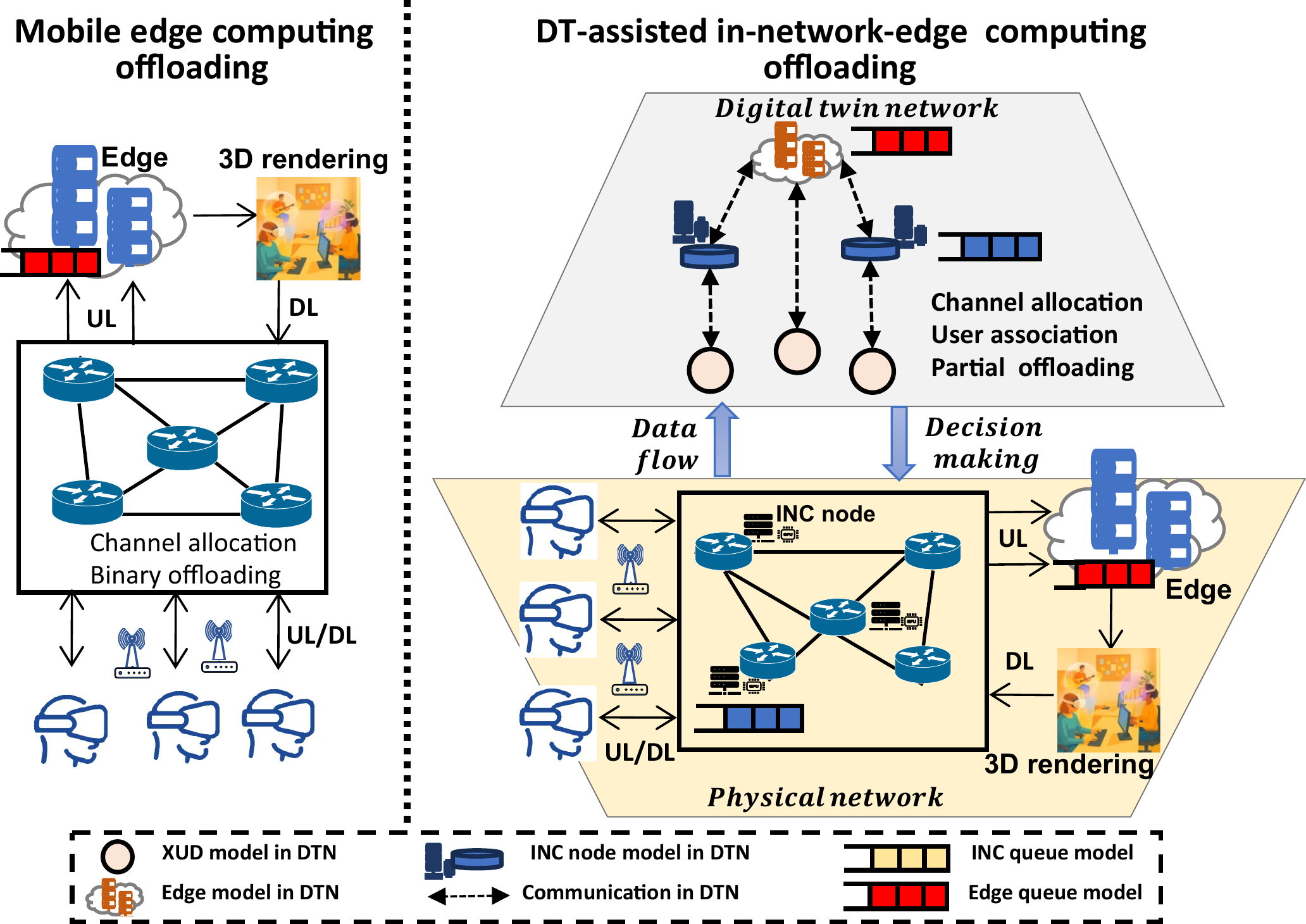}
  \caption{DT-assisted INC-E architecture}
  \label{fig:INC-E_arch}
\end{figure}

\subsection{Our Approach and Rationale}

Our approach addresses three critical challenges for scalable metaverse XR: (i) asymmetric UL/DL computational demands, (ii) decentralized user association under interference, and (iii) joint optimization of partial offloading and downlink energy control. To overcome the limitations of conventional MEC-only architectures, we propose a DT-based INC-E system that integrates INC nodes with traditional edge servers. This architecture enables a progressive inference workflow by functionally splitting task processing across multiple network layers.
The framework is modeled as a Stackelberg game, with the network operator as the leader and XUDs as followers. The operator employs MARL to optimize the offloading-mode (OFMO) and DL power allocation (POAL). OFMO determines whether a task should be processed collaboratively via INC and edge servers or handled solely by the MEC, effectively deciding when INC preprocessing is advantageous or necessary.
Following the operator’s decision, each XUD selects a wireless channel and a partial offloading ratio if the predicted mode is collaborative. Otherwise, the entire task is offloaded directly to the MEC server. This decentralized, interference-aware channel selection implicitly determines each user’s association with a specific INC node, enabling self-organized uplink coordination without centralized control.

\subsection{Motivation and Our Contribution}

Integrating INC-E into metaverse task offloading introduces challenges beyond traditional MEC. Although DT-assisted MEC brings computation closer to users, it still faces resource bottlenecks, latency issues, and limited adaptability in dynamic XR environments. INC-E reduces MEC load and delay through partial in-network execution, but increases the complexity of distributed decision-making and introduces load fluctuations. Consequently, a new adaptive algorithm is required for real-time resource and energy allocation, dynamic OFMO adjustment, and efficient user-to-INC association.
Jointly optimizing task offloading, user association, task splitting, and DL allocation is further complicated by asymmetric UL/DL demands and dynamic channel interference  \cite{dvorovzvnak2020monster, yan2021pricing, li2022digital}. Existing methods often overlook multi-stage in-network processing and struggle with sparse user requests, resulting in inefficient outcomes  \cite{masaracchia2022digital, guo2023intelligent}. To address this limitation, we propose a DT-based Nash multi-agent reinforcement learning (Nash-MARL) framework for interference-aware channel selection and resource allocation. XUD users perform game-theoretic channel selection guided by periodic predictions of operator resource allocation, while XR devices autonomously select their INC nodes with operator assistance to minimize interference. This decentralized approach achieves scalable equilibrium without the need for global coordination  \cite{tutuncuoglu2024sample, hou2022device, yang2023decentralized}. 

The main contributions of this study are summarized as follows:
\begin{itemize}
    \item	We formulate a DT-based system-wide utility maximization problem that jointly integrates task OFMO, user association, partial task splitting, and downlink power control.
    \item We model the interaction between the operator and XUDs as a Stackelberg Markov game, in which operator agents optimize UL and DL parameters. At the same time, XUDs respond through a self-organizing association to achieve overall system objectives.
    \item We design UL and DL agents that operate asynchronously, capturing the sequential dependencies and asymmetric data flow inherent in XR tasks within INC-E environments. 
    \item Comprehensive experiments show our method improves utility, uplink rates, energy efficiency, fairness, and reduces latency for XR, while enabling efficient operator resource allocation over conventional baselines.
\end{itemize}

The remainder of this research is organized as follows: Section~\ref{systemModel} details the system model and problem formulation; Section~\ref{ProposedSolution} presents the proposed solution addressing both user and operator objectives; Section~\ref{NumericalResults} reports numerical results; Section~\ref{RelatedWork} reviews related work; and Section~\ref{Conclusion} concludes the study.

A preliminary version of this work appeared at the IEEE NOMS 2024 Workshop \cite{aliyu2024towards} , presenting partial results under simplified DT settings. The current research extends this study to a broader scenario, including dynamic user association, asynchronous UL and DL optimization, and energy-aware decision-making in INC-E environments. Additionally, we develop a decentralized Nash-AMRL framework to provide a more comprehensive model and evaluation.

\section{System model and problem formulation}
\label{systemModel}
As illustrated in Fig.~\ref{fig:INC-E_arch}, the system comprises XUDs, INC-enabled nodes (CNs), and a MEC server (ES), coordinated by a DT network that facilitates user association, partial offloading, and UL/DL decision-making. The infrastructure supports real-time DT operation via optimized resource allocation, while URLLC links ensure high reliability and minimal latency.
Let $\mathcal{M} = \{1, 2, \ldots, M\}$ denote the set of $M$ user equipments (XUDs), $\mathcal{K} = \{1, 2, \ldots, K\}$ denote the set of $K$ INC nodes (CNs), and $R$ be the ES. The CNs and ESs are connected to an access point (AP) to provide network access for XUDs and APs. 

\subsubsection{Offloading Model in INC-E Network} 
We adopt a time slot model in which XUDs and CNs are fixed within each time slot but may vary across slots. At each time slot $t$, each XUD generates a computational task characterized by $J_m = \{\eta_m, T_m^{\max}\}$, where $\eta_m = \frac{C_m}{I_m^t}$ represents the task complexity (cycles/bits), $I_m^t$ is the task size in bits, $C_m$ is the number of CPU cycles required to execute the task, and $T_m^{\max}$ denotes the maximum tolerable latency for the task $J_m$.
In our scenario, we considered partial offloading, exploiting parallel processing to reduce latency. A portion of a task can be executed at the CNs, while the remaining portion is processed at the MEC server. The MEC server can serve multiple XUDs, whereas CNs have limited capacity.
Let $\Phi_{P,m} \triangleq \{\lambda_{mk}, \aleph_m\}$ denote the partial offloading decision variables, where $\lambda_{mk}$ is the fraction of task $J_m$ executed at INC node $k$, and $\aleph_m = 1 - \sum_{k\in\mathcal{K}} \lambda_{mk}$ is the fraction executed at the edge server (ES). The computing-resource allocation variables are defined as $\Phi_{L,m} \triangleq \{\Phi_{\lambda}, \Phi_{\aleph}\}$, where $\Phi_{\lambda}$ and $\Phi_{\aleph}$ represent the computational resources allocated to the CNs and the ES, respectively. We assume that tasks are generated with high granularity, enabling partial offloading. The task allocation satisfies the constraints: \(
I^t_m = \aleph_m I^t_m + \sum_{k\in \mathcal{K}} \lambda_{mk} I^t_m, \quad C_m = \aleph_m C_m + \sum_{k\in \mathcal{K}} \lambda_{mk} C_m \).

At each time slot $t$, users generate task requests denoted as \( \bm{\mu}_t = \{\mu_1^{(t)}, \mu_2^{(t)}, \ldots, \mu_M^{(t)}\} \). The request state  for user \( m\) can be \( 0 \) (no request) or \( f \in \mathcal{F} \) (requesting task \( f \)). These requests follow a first-order Markov chain with \((|\mathcal{F}| + 1)\) states, where future requests depend only on the current state, and the transition probabilities are unknown \cite{chen2022dynamic}.

Let $\mathbf{u}^t = \left( u_{mk}^t \right)_{m \in \mathcal{M}, k \in \mathcal{K}}$  denote the collection of user associations, which indicates the number of XUDs that CN can support in a given time slot $t$. 
Let the XUD channel selection decision variable be $s_m \in \{0,1,\ldots,K\}$, where $s_m=0$  represents a non-collaborative case in which the task is offloaded to the MEC channel, and $s_m=k$ ($k\in\mathcal{K}$) denotes that the XUD selects channel $k$, corresponding to its associated CN $k$.
We denote the predictive OFMO state as$o_m \in \{0,1\}$, where $o_m=0$ indicates non-collaboration and forces $s_m=0$, whereas $o_m=1$  indicates collaboration and forces $s_m=1$. Formally, this logical relation is captured by the coupling constraint  $o_m^t = 0 \;\Rightarrow\; s_m^t = 0$, equivalent
\begin{equation}
 s_m^t (1 - o_m^t) = 0.
\label{eq:ofmo_coupling}
\end{equation}
Thus, when  $o_m^t = 0$, the only feasible choice is  $s_m^t = 0$, whereas any collaborative selection   implicitly requires $o_m^t = 1$.

\subsubsection{INC-E DT Model}
DT services generate virtual replicas of physical systems, replicating hardware, applications, and real-time data. The URLLC-based INC-E’s DT is defined as $\ DT = \{\tilde{\mathcal{M}},\ \tilde{\Phi}_L\},$ where $\{\tilde{\mathcal{M}},\ \tilde{\Phi}_L\}$ represents the system’s virtual mirror, including $M$ XUDs and $\tilde{\Phi}_L$ (CNs and ES). The DT layer, informed in real time, automates control through services such as data analysis, decision-making, and instant optimization, focusing on tasks such as offloading strategies and resource allocation.

Each $m$-th XUD’s specific DT is associated with a CN node
for processing and defined as $\ DT_m^{cn} = (f_m^{cn}, \tilde{f}_m^{cn}),$ where $f_m^{cn}$ 
denotes the estimated processing rate, and  $\tilde{f}_m^{cn}$ quantifies the deviation from the actual processing rate between the physical XUD and its DT  \cite{van2022urllc}. In the DT layer, the estimated processing rate,  $f_m^{cn}$, mirrors XUD behavior, driving optimization decisions for device configurations. This rate is the focus of our optimization, with its deviation set as a predetermined percentage for simulations, following established practices \cite{do2022digital}.

Similarly, for the $\tilde{\Phi}_L$-th INC-E computing resource (CNs
and	ES),	its	DT	($DT_{\tilde{\Phi}_L}^{cn}$)	is	formulated	as $DT_{\tilde{\Phi}_L}^{cn} = (f_{\tilde{\Phi}_L}^{cm}, \tilde{f}_{\tilde{\Phi}_L}^{cm}),$ 
where $f_{\tilde{\Phi}_L}^{cm}$ represents the estimated processing rate of the real INC-E, and $\tilde{f}_{\tilde{\Phi}_L}^{cm}$ characterizes the discrepancy in the processing rate estimate. The DT emulation of INC-E (CNs and ES) provides valuable insights into processing rates, facilitating efficient resource allocation and reducing latency by adjusting offloading ratios and computing resource allocation.

\subsection{Communication Model}
The AP, equipped with $L$ antennas and serving $M$ single-antenna XUDs, establishes channel connections with computing resource $\tilde{\Phi}_L$ represented by $\bm{h}_{m\tilde{\Phi}_L} = \sqrt{g_{m\tilde{\Phi}_L}} \cdot \bm{ \bar{h}}_{m\tilde{\Phi}_L}$ , where $g_m$ is the large-scale channel coefficient and $\bm{ \bar{h}}_{m\tilde{\Phi}_L}$ is small-scale fading following $CN(0, \bm{I})$, with $CN(.,.)$ denoting a complex circularly symmetric Gaussian distribution. A channel matrix $\bm{H}_{\tilde{\Phi}_L} = [\bm{h}_{1\tilde{\Phi}_L}, \bm{h}_{2\tilde{\Phi}_L}, \ldots, \bm{h}_{M\tilde{\Phi}_L}] \in \mathbb{C}^{L \times M}$ contains connections from $m$-th SM to the $\tilde{\Phi}_L$-th AP. Each XUD is allocated bandwidth, $b_m$. Match filtering and successive interference cancellation (MF-SIC) are employed to improve transmission performance \cite{fang2001performance}.

\subsubsection{Uplink}
The signal-to-interference-plus-noise (SINR) at the $\tilde{\Phi}_L$-th AP by the $m$-th XUD is defined as $\gamma^t_{m\tilde{\Phi}_L}(\mathbf{p^t},\mathbf{u}^t) = \frac{\sum_{k \in \mathcal{K}} u_{mk}^t p^t_{m\tilde{\Phi}_L} \lVert h^t_{m\tilde{\Phi}_L} \rVert^2}{\mathcal{I}^t_{m\tilde{\Phi}_L}(\mathbf{p^t},\mathbf{u}^t) + N_0},$ where $p^t_{m\tilde{\Phi}_L}$ is the transmit power of the $m$-th XUD, $N_0$ is the noise power, $p^t = [p^t_{m\tilde{\Phi}_L}]_{m=1}^M$, and $\mathcal{I}^t_{m\tilde{\Phi}_L}(\mathbf{p^t},\mathbf{u}^t) = \sum_{n \neq m}^{M} p^t_{n\Phi_L} \frac{\lvert h_{m\Phi_L}^H h_{n\Phi_L} \rvert^2}{\lVert h_{m\Phi_L} \rVert^2} $ is the interference by other XUDs. Consequently, the uplink URLLC transmission rate is expressed as \cite{ren2020joint,she2017radio}:
\begin{equation}
\label{eq:uplink_rate}
\begin{aligned}
\omega^t_{m\Phi_L}(\mathbf{p}^t,\mathbf{u}^t)
&\approx B \log_2\!\left(1 + \gamma_{m\Phi_L}(\mathbf{p}^t,\mathbf{u}^t)\right) \\
&\quad - B \sqrt{\frac{V_{m\Phi_L}(\mathbf{p}^t,\mathbf{u}^t)}{N}}
\frac{Q^{-1}(\epsilon)}{\ln 2}.
\end{aligned}
\end{equation}

\noindent where $B$ is the system bandwidth, $\epsilon$ denotes the decoding error probability, $\gamma_{m\tilde{\Phi}_L}(p,n)$ is the signal-to-noise ratio (SNR) observed by the m-th XUD,
$Q^{-1}(.)$ is the inverse Q-function, and $Q(x) = \frac{1}{\sqrt{2\pi}} \int_x^{\infty} e^{-t^2/2} \,dt$, and $V_{m\tilde{\Phi}_L}$ is the channel dispersion, defined as $V_{m\tilde{\Phi}_L}(p,n) = 1 - \left[1 + \gamma_{m\tilde{\Phi}_L}(\mathbf{p^t},\mathbf{u}^t)\right]^{-2}$. This expression calculates the uplink
rate for the selected destination, accounting for channel conditions, bandwidth allocation, and transmit power. In this study, we do not optimize the uplink power.

The minimum transmission rate of user k is determined by \eqref{eq:uplink_rate}, with the interference from other users assumed to remain below a threshold ${\mathcal{I}^t_{m\tilde{\Phi}_L}(\mathbf{p^t},\mathbf{u}^t) + N_0}$ . Following \cite{chen2022dynamic}, our algorithm manages channel interference without considering power control or specific coding schemes. When multiple users offload tasks to the same location, interference can considerably reduce uplink rates, forcing some users to process tasks locally.

From \eqref{eq:uplink_rate}, as multiple users select the same INC node (or channel), interference grows, decreasing uplink rates and increasing latency. As a result, users self-organize to avoid overloaded INCs, achieving a decentralized, interference-aware association.
The UL wireless transmission latency from the $m$-th XUD to the $\tilde{\Phi}_L$-th INC-E resource is expressed as:
\begin{equation}
    T_{m\tilde{\Phi}_L}^{\text{CO}}(\mathbf{p^t},\mathbf{u}^t,\tilde{\Phi}_L) = \max_{\forall\Phi_L} \left\{\frac{\sum_{k \in \mathcal{K}} u_{mk}^t \Phi_L I^t_m}{\omega_{m\Phi_L}(\mathbf{p^t},\mathbf{u}^t)}\right\}. \label{eq:uplink_latency}
\end{equation}

\noindent where $p_m^t$ is the uplink transmit power, $h_m^t$ is the channel gain, and $\sigma^2$ is the Gaussian noise variance. This formulation inherently accounts for channel interference, guiding OFMO decisions through a communication-performance trade-off.

\subsubsection{Downlink}
Let let $\mathbf{p'^t}= (p'^t)_{m \in \mathcal{M}}$  denote the DL POAL used by the MEC server (ES) to communicate with XUDs in time slot $t$.
Unlike the uplink transmit power $\mathbf{p^t}$  at the UL stage, the $\mathbf{p'^t}$ is optimized at the DL stage. Since INC nodes are not involved in the downlink, the transmission and energy formulations depend solely on the MEC server. Accordingly, the DL URLLC transmission rate is expressed as:
\begin{equation}
\begin{split}
    \omega'^t_{m\Phi_L}(\mathbf{p'^t},\mathbf{u}^t) &\approx B \log_2\left[1 + \gamma^t_{m\Phi_L}(\mathbf{p'^t},\mathbf{u}^t)\right] \\
    &\quad - B\sqrt{\frac{V_{m\Phi_L}(\mathbf{p'^t},\mathbf{u}^t)}{N}} \frac{Q^{-1}(\epsilon)}{\ln 2}, \label{eq:downlink_rate}
\end{split}
\end{equation}

In the DL stage, the data $I^t_m$ received in the UL is digitized at the MEC server and returned to the XUD. The rendered data is expressed as \cite{yu2023asynchronous}:

\begin{equation}
    I'^t_m = f_m(I^t_m)
\end{equation}

\noindent where $f_m(\cdot)$ is the size translation and rendering function from 2D to 3D, tailored to the specific metaverse application. Accordingly, the DL transmission latency is expressed as
\begin{equation}
    T_{m}^{'\text{CO}}(\mathbf{p'^t},\mathbf{u}^t) = \max_{\forall\Phi_L} \left\{\frac{\Phi_L I'^t_m}{\omega'^t_{m\Phi_L}(\mathbf{p'^t},\mathbf{u}^t)}\right\}. \label{eq:downlink_latency}
\end{equation}
The DL energy consumed at time step $t$ is given by
\begin{equation}
    E'^t(p'^t_m) = \sum_{m \in \mathcal{M}} p'^t_m \times \omega'^t_{m\Phi_L}(\mathbf{p'^t},\mathbf{u}^t).
\end{equation}

\subsection{Computation Model}
In the computation model, each XUD generates a granular computation task $J_m$, a portion of which can be executed by the CNs and the remaining portion at the ES. The model is defined as follows:

\subsubsection{Non-collaborative task offloading}
In the non-collaborative mode, particularly relevant for metaverse XR scenarios, users offload tasks entirely to the MEC server for computation and subsequent rendering of 3D data. This approach follows conventional offloading methods reported in the literature \cite{yu2023asynchronous}. 
The task $J_m$ when the offloaded portion satisfies $\aleph_m=1$, it is executed at the ES with the estimated processing rate $f_m^{em}$, incurring a latency of: $\tilde{T}_m^{em}(\aleph_m,f_m^{em}) = \frac{\aleph_m C_m}{f_m^{em}}$. The latency gap
$\Delta T_m^{em}$ between the actual execution and the DT estimated $\Delta T_m^{em}(\aleph_m,f_m^{em}) = \frac{\aleph_m C_m \tilde{f}_m^{em}}{f_m^{em}(\tilde{f}_m^{em}-f_m^{em})}$. Thus, the actual execution latency at the MEC is 	$T_m^{em}(\aleph_m,f_m^{em}) = \Delta T_m^{em}(\aleph_m,f_m^{em}) + \tilde{T}_m^{em}(\aleph_m,f_m^{em})$.
The total latency in non-collaborative mode comprises the uplink transmission delay $T_{m\tilde{\Phi}_{\aleph}}^{\text{CO}}(\mathbf{p^t},\mathbf{u}^t,\tilde{\Phi}_{\aleph})$, MEC execution delay $T_m^{em}$, and downlink delay $T_{m\tilde{\Phi}_{\aleph}}^{'\text{CO}}(\mathbf{p'^t},\mathbf{u}^t,\tilde{\Phi}_{\aleph})$ expressed as:  
\begin{equation}
T_m^{O} = T_{m\tilde{\Phi}_{\aleph}}^{\text{CO}}(\mathbf{p^t},\mathbf{u}^t,\tilde{\Phi}_{\aleph}) + T_m^{em}(\aleph_m,f_m^{em}) + T_{m\tilde{\Phi}_{\aleph}}^{'\text{CO}}(\mathbf{p'^t},\mathbf{u}^t,\tilde{\Phi}_\aleph).
\end{equation}

\subsubsection{Collaborative task offloading}: In the collaborative mode, users partially offload their tasks to INC nodes for preprocessing, with the final rendering or aggregation performed at the MEC server before returning results to the users. 
For the INC node, the portion $J_m$ portion $\lambda_{mk}$ is executed at the CNs with the estimated processing rate $f_m^{cn}$. Consequently, the estimated CN execution latency is $\tilde{T}_{mk}^{cn}( \lambda_{mk},f_m^{cn}) = \max_{\forall k \in \mathcal{K}} \left\{\frac{\lambda_{mk} C_m}{f_m^{cn}}\right\} \label{eq:INC_processing_time}$. The discrepancy
between the real-world and DT execution at the CNs is estimated as:	 $\Delta T_{mk}^{cn}(\lambda_{mk},f_m^{cn}) = \frac{\lambda_{mk} C_m \tilde{f}_{mk}^{cn}}{f_m^{cn} (f_m^{cn} - \tilde{f}_{mk}^{cn})}$. Thus, the actual CN
processing time is $T_{mk}^{cn}\!\left(\lambda_{mk},f_m^{cn}\right)
= \Delta T_{mk}^{cn}\!\left(\lambda_{mk},f_m^{cn}\right) + \tilde{T}_{mk}^{cn}\!\left(\lambda_{mk},f_m^{cn}\right)$.  
 
The total latency in the collaborative mode, including transmission to the INC $T_{m\tilde{\Phi}_{\lambda}}^{\text{CO}}(\mathbf{p^t},\mathbf{u}^t,\tilde{\Phi}_{\lambda})$, INC preprocessing time $T_{mk}^{cn}(\lambda_{mk},f_m^{cn})$, MEC execution and downlink delay, $T_{m\tilde{\Phi}_{\lambda}}^{'\text{CO}}(\mathbf{p'^t},\mathbf{u}^t,\tilde{\Phi}_{\lambda})$, is given by
\begin{equation}
\label{eq:collab_latency}
\begin{aligned}
T_m^{C}
&= T_{m\tilde{\Phi}_{\lambda}}^{\text{CO}}
   \!\left(\mathbf{p}^t,\mathbf{u}^t,\tilde{\Phi}_{\lambda}\right) + T_{mk}^{cn}
   \!\left(\lambda_{mk},f_m^{cn}\right) \\
&\quad + T_{m\tilde{\Phi}_{\lambda}}^{\prime\text{CO}}
   \!\left(\mathbf{p}^{\prime t},\mathbf{u}^t,\tilde{\Phi}_{\lambda}\right).
\end{aligned}
\end{equation}

\subsubsection{Latency Model}
The end-to-end (E2E) latency $J_m$ of task $J_m$ comprises transmission, processing, queueing, and forwarding delays accumulated across the INC node and the MEC server, depending on the operator-selected offloading mode.
When $o_m = 0$ the entire task fraction $\aleph_m = 1$ is executed at the MEC server. The corresponding E2E latency is $\widetilde{T}_m^{O} =
T_m^{O} + Q^{\mathrm{mec}}$, where Qmec denotes the MEC queueing delay. When $o_m = 1$, a task fraction $\lambda_{mk}$ is executed at INC node $k$, while $\aleph_m = 1 - \sum_{k}\lambda_{mk}$ is executed at the MEC server. The corresponding end-to-end latency is  $\widetilde{T}_m^{C} = T_m^{C} + Q^{\mathrm{cn}}_{k_m} + T_m^{\mathrm{fw}}(k_m) + Q^{\mathrm{mec}}$, where $Q^{\mathrm{cn}}_{k_m}$ is the queueing delay at the selected INC node $k$, $T_m^{\mathrm{fw}}(k_m)$ is the forwarding delay from the INC node $k$ to the MEC, and $Q^{\mathrm{mec}}$ is the MEC queueing delay. The overall latency under the operator-selected offloading mode is
	\begin{equation}
T_m^{\mathrm{e2e}}
=
\mathds{1}_{\{s_m = 0\}}\, \widetilde{T}_m^{O}
+
\mathds{1}_{\{s_m \neq 0\}}\, \widetilde{T}_m^{C}.
\label{eq:unified_e2e_final}
\end{equation}
     
Note that this latency $T_m^{\mathrm{e2e}}$ depends solely on the user’s offloading decision
 $s_m^t$. The operator variable $o_m^t$ acts only as a feasibility control via the coupling constraint $s_m^t(1 - o_m^t)=0$, and therefore does not appear explicitly in~\eqref{eq:unified_e2e_final}.

\subsection{Problem Formulation}
We assume consider a system that in which each XUD $m \in \mathcal{M}$ autonomously determines its offloading strategy so as to minimise minimize task latency and energy consumption while maximising maximizing its own utility, given the operator’s OFMO decision. The instantaneous utility of XUD $m$ at time slot $t$ is then expressed as
\begin{equation}
\begin{aligned}
U_m^t
&=
g_t \Big[
\widetilde{T}_m^{O}
-
T_m^{\mathrm{e2e}}
\Big]
-
w_m E'^t(p'^t_m)  \\
\end{aligned}
\label{eq:utility_final_clean_corrected}
\end{equation}
\noindent where $E'^t(p'^t_m)$ denotes the downlink energy consumed by user $m$ at time $t$.

The XUD seeks aims to maximise maximize its long-term, time-averaged utility:
\begin{align}
\mathcal{P}_U:~\max_{\{s_m^t\}}~~&
\lim_{T\to\infty} 
\frac{1}{T} 
\sum_{t=1}^{T} U_m^t
\\
\text{s.t.}\quad
& s_m^t (1 - o_m^t) = 0,
\tag{12a}
\label{eq:constraint_15a_updated}
\\
& T_m^{\mathrm{e2e}}(s_m^t) \leq T_m^{\max}, 
\tag{12b}
\label{eq:constraint_15b_updated}
\\
& s_m^t \in \{0,1,\ldots,K\}.
\tag{12c}
\label{eq:constraint_15c_updated}
\end{align}

From the operator’s perspective, the instantaneous task latency experienced by XUD $m$ in slot $t$ is given by ~\eqref{eq:unified_e2e_final}. The corresponding instantaneous network cost is
\begin{equation}
C^t_{\pi} 
=
\sum_{m\in\mathcal{M}}
\bigl[
T_m^{\mathrm{e2e}}
+
E'^t(p'^t_m)
\bigr]
\label{operator_cost_updated_final}
\end{equation}
where $E'^t(p'^t_m)$ denotes the downlink energy cost associated with the power allocation $p'^t_m$.

The operator’s policy $\theta_{\pi}$ determines both the OFMO vector $\mathbf{o}^t = (o_1^t,\ldots,o_M^t)$  and the downlink POAL vector $\mathbf{p'}^t = (p'^t_1,\ldots,p'^t_M)$, and aims to minimize the long-term discounted cost:
\begin{align}
\mathcal{P}_O:~\min_{\theta_{\pi}}~~&
\mathbb{E}_{\pi}
\left[
\sum_{t=0}^{\infty} 
\gamma^t C^t_{\pi}
\right]
\\
\text{s.t.}\quad
& p'^t_m \in [p'_{\min}, p'_{\max}], 
\quad \forall m,t,
\tag{14a}
\\
& \sum_{\{m:s_m^t=j\}} I^t_m \leq C_j, 
\quad \forall j\in\mathcal{K},
\tag{14b}
\end{align}
where $I^t_m$ is the input data size of user $m$, and $C_j$ is the capacity of INC node $j$.

The interaction between XR users and the network operator is modeled as a Stackelberg Markov game. In each time slot, the operator (leader) first announces $(\mathbf{o}^t,\mathbf{p'}^t)$ according to its policy $\theta_{\pi}$, and then the XUDs (followers) respond by selecting $s_m^t$ to minimize their utilities under the coupling constraint \eqref{eq:ofmo_coupling}.

Within this framework, address three fundamental questions. First, we examine whether a subgame perfect equilibrium (SPE) exists, i.e., whether there exists a joint offloading–allocation strategy from which neither the XUDs nor the operator has any incentive to deviate at any stage of the game. 

\begin{definition}[SPE]
Let $(\bm{o}^{\ast}, \bm{p'}^{\ast})$ be a solution to the operator’s optimization 
problem, and $\bm{s}^{\ast}$ the best response strategy profile of the XUDs. 
Then the tuple $(\bm{s}^{\ast}, \bm{o}^{\ast}, \bm{p'}^{\ast})$ is a 
subgame perfect equilibrium (SPE) of the hierarchical Stackelberg game $
\Gamma = \{\Gamma_O,\Gamma_U\}$,
if the operator and the XUDs satisfy
\[
C(\bm{s}^{\ast},\bm{o}^{\ast},\bm{p'}^{\ast})
\leq
C(\bm{s}^{\ast},\bm{o},\bm{p'}),
\quad \forall (\bm{o},\bm{p'}),
\]
\[
U_m^t(s_m^{\ast},o_m^{\ast},p'^{\ast}_m)
\geq
U_m^t(s_m,o_m^{\ast},p'^{\ast}_m),
\quad 
\forall m,\; s_m\in\{0,\ldots,K\}.
\]
\end{definition}
If the game admits an SPE, the second question is whether the SPE can be computed efficiently. The third question explores whether a MARL-based learning framework can be applied to learn optimal policies for the operator while simultaneously maximizing the  user utility. Before addressing these questions, we first characterize the XUDs’ best response and provide analytical results under simplifying assumptions.

\noindent\textbf{Definition 2 (Pure Nash Equilibrium and Best Reply).}
A pure strategy Nash equilibrium (NE) is a strategy profile 
$\bm{s}^{\ast} = (s_m^{\ast}, s_{-m}^{\ast})$
in which every XUD plays a best reply to the strategies chosen by all other users. 
Formally, $\bm{s}^{\ast}$ is a NE if
\[
U_m(s_m^{\ast}, s_{-m}^{\ast}) 
\ge 
U_m(s_m, s_{-m}^{\ast}), 
\qquad 
\forall s_m \in \mathcal{S}_m,\ \forall m \in \mathcal{M}.
\]
\noindent
Given a strategy profile $(s_m, s_{-m})$, a \emph{better reply} of XUD $m$
is any strategy $s'_m$ such that
\[
U_m(s'_m, s_{-m})
>
U_m(s_m, s_{-m}),
\]
Similarly \emph{best reply} is a strategy 
$s_m^{\ast}$ satisfying
\[
U_m(s_m^{\ast}, s_{-m})
\ge 
U_m(s_m, s_{-m}),
\qquad 
\forall s_m \in \mathcal{S}_m.
\]

\section{Proposed solution}
\label{ProposedSolution}
The objective function is nonconvex due to the presence of partial offloading decision variables and nonlinear relationships. In practice, it is computationally intractable to solve the problem directly, as the overall decision-making process involves binary channel selection and OFMO variables, along with complex latency-energy interactions over time. To address this challenge, we decompose the DT optimization problem into two subproblems: interference-aware channel selection and resource allocation (OFMO/POAL) across time slots in the INC-E cyber twin, without XUD request transition probabilities. Fig.~\ref{fig:INC-E architecture} illustrates the proposed solution framework, which uses game theory for channel selection and AMRL for resource allocation. Detailed descriptions of the proposed methods are provided in the following subsection.

\begin{figure}[ht]
  \centering
  \includegraphics[width=0.35\textwidth]{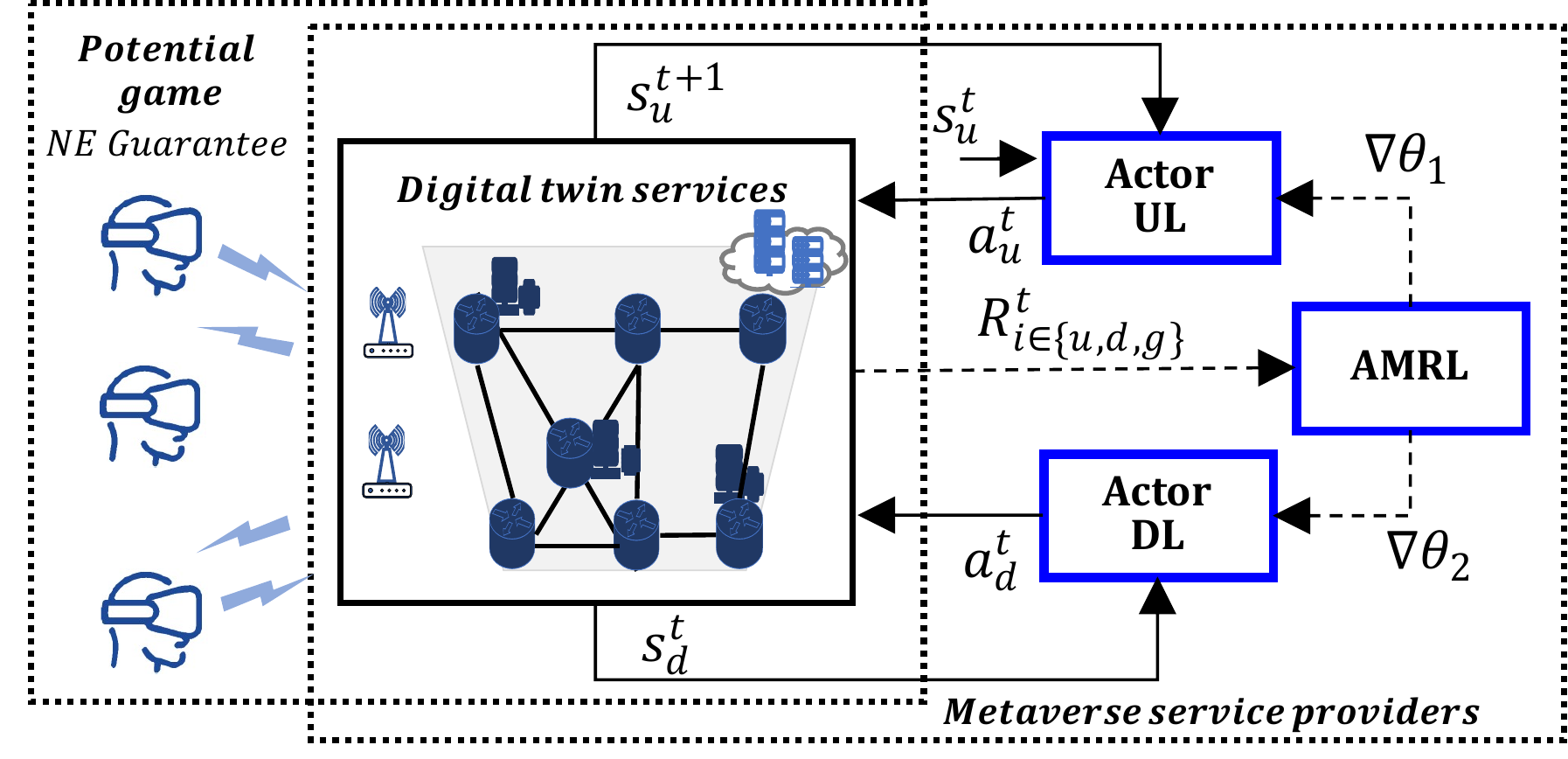}
  \caption{Nash-AMRL scheme for DT-assisted INC-E system}
  \label{fig:INC-E architecture}
\end{figure}

\begin{algorithm}[t]
\caption{Best-Response Algorithm for XUD Association and Task Splitting}
\label{alg:epg_user}
\begin{algorithmic}[1]
\State Initialize $s_m\leftarrow 0$, $\forall m\in\mathcal{M}$
\Repeat
    \State Receive $(\mathbf{o}^t,\mathbf{p'}^t)$ from the operator
    \For{each XUD $m\in\mathcal{M}$}
        \If{$o_m^t=0$}
            \State $s_m\leftarrow 0$
        \Else
            \State Compute $\mathcal{I}_m^t$ and $\omega_{m\Phi_L}^t$ from \eqref{eq:uplink_rate}
            \State Construct
            \[
            \mathcal{S}_m^t=\{k\in\{1,\ldots,K\}: T_m^{\mathrm{e2e}}(k)\le T_m^{\max}\}
            \]
            using \eqref{eq:unified_e2e_final}
            \For{each $k\in\mathcal{S}_m^t$}
                \State Solve $\mathcal{P}_{\widetilde U}$ for $(\lambda_{mk}^\star,\beta_m^\star)$ via \eqref{eq:lambda_update}--\eqref{eq:nu_update}
                \State Compute $U_m(k,s_{-m};\lambda_{mk}^\star,\beta_m^\star)$
            \EndFor
            \State $s_m^\star\leftarrow \arg\max_{k\in\mathcal{S}_m^t} U_m(k,s_{-m};\lambda_{mk}^\star,\beta_m^\star)$
            \If{$s_m^\star\neq s_m$}
                \State $s_m\leftarrow s_m^\star$
            \EndIf
        \EndIf
    \EndFor
\Until{no XUD changes its strategy}
\State \Return $(\mathbf{s}^\star,\lambda^\star,\beta^\star)$
\end{algorithmic}
\end{algorithm}

\subsection{ Multi-user interference-aware channel selection game}

Given the operator-predicted OFMO in each time slot, XR user devices (XUDs) act as followers within a Stackelberg framework. In MEC-only mode, all tasks are transmitted directly to the MEC server. In collaborative mode, XUDs select a wireless channel for each mapped INC node and may additionally determine a partial offloading ratio. UL interference, queueing delays, and INC processing loads depend on all collective choices of all users, thereby forming a decentralized interference-aware channel-selection game.

Let each XUD $m$ select a strategy $s_m = k_m \in \{0,1,\ldots,K\}$ where $k_m \in \{0,1,\ldots,K\}$ indicates the selected uplink channel (and associated
INC node), with $k_m=0$ corresponding to MEC-only offloading.

The operator does not participate in this game; it only determines whether the game exists by predicting the OFMO. When $o_m = \mathrm{MEC\!-\!only}$, XUD $m$ is excluded from the game.

For a user selecting channel $k_m$, the uplink SINR is $\gamma^t_{m\tilde{\Phi}_L}(\mathbf{p}^t,\mathbf{u}^t)$
yielding the finite-blocklength URLLC rate $\omega_{m\tilde{\Phi}_L}^t$. The utility function $U_m(s_m,s_{-m})$ follows the operator-assisted design and depends parametrically on the current OFMO and POAL decisions. Each XUD seeks the best response $s_m^* = \arg\max_{s_m} U_m(s_m,s_{-m})$ given the operator’s OFMO and the channel selections of other users.

From the SINR expression, the effective interference on channel $k$ is $\mathcal{I}_{m}(k)
=
\sum_{n \neq m:\, k_n = k}
p_n 
\frac{|h_{m}^H h_{n}|^2}{\|h_m\|^2}$. A larger $\mathcal{I}_{m}(k)$ decreases $\omega_{m\tilde{\Phi}_L}$ and increases both the UL latency and the INC queue. Thus XUD $m$ prefers channel $k$ only if
\begin{equation}
T_m^{\mathrm{e2e}}(k) < T_m^{\mathrm{non}}
\quad\Longleftrightarrow\quad 
\mathcal{I}_m(k) \le V_m,
\label{equ:interf}
\end{equation}

\noindent $V_m$ denotes a user-specific interference tolerance threshold derived from latency constraint $T_m^{max}$ is satisfied. Hence users self-organize away from overloaded channels, achieving decentralized interference-aware INC association without global coordination.

In other words, for a given task channel-selection strategy $s_m$, the XUD $k$ can increase system utility  when its received interferences satisfy inequality  ~\eqref{equ:interf}. The channel-selection interaction forms a discrete game
$G = \{\mathcal{M}, (S_m)_{m\in\mathcal{M}}, (U_m)_{m\in\mathcal{M}}\}$,
where $\mathcal{M}$ is the set of players (XUDs), $S_m$ is the strategy
space of XUD $m$, and $U_m$ is its computing utility. Based on \cite{pham2022partial}, the game $G$ is an exact potential game(EPG) with potential function $\phi(s)$ formulated as follows:

\begin{equation}
\begin{split}
\phi(s) &= \sum_{m\in\mathcal{M}} \Big[
s_{m0} R_{m0} \\
&\quad + (1 - s_{m0})\big(
\sum_{j\in\mathcal{K}} s_{mj} R_{mj}
+ \sum_{m'\neq m} R_{m'0}
\big)
\Big]
\end{split}
\end{equation}

where $R_{mj} =
g_t \Big[
\widetilde{T}_m^{O}
-
T_m^{\mathrm{e2e}}
\Big]
-
w_m E'^t(p'^t_m)$.
For ease of proof, the expression $\phi(s_m,s_{(-m)})$ is given as:
\begin{equation}
\begin{split}
\phi(s_m,s_{(-m)}) &= \sum_{m\in\mathcal{M}} \Big[
s_{m0} R_{m0} \\
&\quad + (1 - s_{m0})\big(
\sum_{j\in\mathcal{K}} s_{mj} R_{mj}
+ \sum_{m'\neq m} R_{m'0}
\big)
\Big]
\end{split}
\end{equation}

\textit{Remark 1:} The game $G$, with potential function $\phi(s)$, is an EPG and can reach a NE in a finite number of iterations.

Given a fixed channel/INC selection $k_m$ and interference profile, XUD $m$ locally chooses its partial offloading ratio $\lambda_{mk}\in \Phi_{P,m}$ and requested INC computing share $\beta_m$ to maximize utility under latency constraints. The operator enforces resource feasibility, while the utility function captures the user-side latency–energy trade-off guiding optimal task-splitting. Accordingly, XUD $m$ adjusts $(\lambda_{mk},\beta_m)$ to maximize utility subject to latency and resource constraints:

\begin{align}
\mathcal{P}_{\widetilde{U}}:~\max_{\lambda_{mk},\,\beta_m}\quad
& U_m(\lambda_{mk},\beta_m; s_m)
\label{prob:P1_obj}
\\[2pt]
\text{s.t.}\quad
& T_m^{\mathrm{e2e}}(\lambda_{mk},\beta_m; s_m)
\;\le\;
T_m^{\max}, \tag{18a}
\label{prob:P1_delay_constraint}
\\
& 0 \;\le\; \lambda_{mk} \;\le\; 1,
\quad
0 \;\le\; \beta_m \;\le\; 1. \tag{18b}
\label{prob:P1_box_constraints}
\end{align}

Under the latency model in~\eqref{eq:unified_e2e_final}, 
$T_m^{\mathrm{e2e}}(\lambda_{mk},\beta_m)$ is monotonically increasing and convex 
with the load variables, whereas $U_m(\lambda_{mk},\beta_m)$ is concave in 
$(\lambda_{mk},\beta_m)$ for a fixed channel selection $s_m$.
Hence, $\mathcal{P}_{\widetilde{U}}$ is a convex program that can be solved via
Lagrangian relaxation:
\[
\mathcal{L}_m
=
U_m
+ \theta_m\!\left(T_m^{\mathrm{e2e}} - T_{\max}\right)
+ \nu_m(\beta_m - 1),
\]
with primal updates

\begin{align}
\lambda_{mk}^{(t+1)}
&=
\left[
\lambda_{mk}^{(t)}
- \kappa_1
\frac{\partial\mathcal{L}_m}{\partial \lambda_{mk}}
\right]_{[0,1]}, \nonumber
\\
\beta_m^{(t+1)} 
&=
\left[
\beta_m^{(t)}
- \kappa_2
\frac{\partial\mathcal{L}_m}{\partial \beta_m}
\right]_{[0,1]},
\label{eq:lambda_update}
\end{align}

and dual updates
\begin{align}
\theta_m^{(t+1)}
&=
\left[
\theta_m^{(t)}
+ \kappa_3
\left(
T_m^{\mathrm{e2e}} - T_{\max}
\right)
\right]_{+},\nonumber
\\
\nu_m^{(t+1)} 
&=
\left[
\nu_m^{(t)}
+ \kappa_4
\left(
\beta_m - 1
\right)
\right]_{+}.
\label{eq:nu_update}
\end{align}

Equivalently, one may define a per-user cost function: 
$\mathcal{G}_m(\lambda_{mk},\beta_m; s_m) = -U_m(\lambda_{mk},\beta_m; s_m)$ 
and solve a minimization problem as in \cite{pham2022partial}. In this study, we adopt the utility-maximization form for consistency with our
game-theoretic framework

\textit{Remark 2 :}
 The convex continuous optimization preserves the equilibrium of the discrete game $G$ while improving user utility, consistent with  \cite{pham2022partial}.

 Based on Remark 1 and 2, we design a potential game-based Multi-user ICS algorithm to address the NE solution for all users. The detailed steps of ICS algorithm are summarized in Algorithm 1

\subsection{Markov decision process (MDP) formulation for OFMO and POAL}

Given the XUDs’ best response, we now optimize network parameters while accounting for stochastic user demands and dynamic network conditions. Limited INC node capacity and variable user demand imply that the OFMO policy should be guided by predicted user requests and resource availability. 
The operator’s problem is modeled as an asynchronous sequential decision-making process under uncertainty, with OFMO and POAL occurring in the UL and DL processes, respectively. We reformulate this problem into asynchronous UL and DL processes and employ AMRL online learning to approximate the optimal policies. 

\subsubsection{State} To mitigate erratic training caused by complex states, we consider only relevant variables that exhibit collection or sequential attributes, as follows:
\begin{enumerate}
    \item Uplink state $X_{u,m}^t$ - user task request, task input size, transmission power:
    \( X_{u,m}^t = (\mu_m, I^t_m, p_m)\)
    \item Downlink state $X_{d,m}^t$ - UL actions (user association), resultant data size from 3D rendering, task input size: \( X_{d,m}^t = (\mathbf{u}_m^t, I'^t_m, I^t_m) \)
    
\end{enumerate}
We denote, \( X_m^t = (X_{u,m}^t, X_{d,m}^t) \in \mathcal{X} \) be the local network state at XUD \(m \in \mathcal{M}\). Thus, the global state
\( X^t = (X_m^t)_{m \in \mathcal{M}} \in \mathcal{X}^{\mathcal{M}} \)
denotes the global state of the network.

\subsubsection{Action} The action of the UL agent ($\pi_u^t$) and  DL agent ($\pi_d^t$) are defined as follows:

\begin{enumerate}
    \item Uplink action $\pi_u^t$ - the action involves OFMO preference scores used to infer a feasible OFMO $(\mathbf{u}^{t+1})$. The agent’s objective is to balance the operator’s load at the INC: \(\pi_u^t=\{\mathbf{u}^t_k, \mathbf{u}^t_{1+k}\}  \forall k \in \mathcal{K}\)
    \item Downlink action $\pi_d^t$ - the action corresponds to determining the continuous POAL for each XUD: \(\pi_d^t=\{\mathbf{p'}_1^t, \mathbf{p'}_{2}^t, \dots \mathbf{p'}_M^t\}   \forall \mathbf{p'}^t \in [\mathbf{p'}_{min}, \mathbf{p'}_{max}]\)
\end{enumerate}

Thus, the operator's policy in a given time slot is denoted as \(\pi=\{\pi_u^t, \pi_u^t\}\).

\subsubsection{Reward} In our scenario, the operator aims to balance load and assign optimal power in an asynchronous system with asymmetric data sizes. Although each agent has distinct objectives, a global target exists that aligns the interests of both agents. Agent $1$ UL focuses on balancing the load based on the XUD game outcome while monitoring the UL data rate to ensure the DL agent can effectively manage the uploaded data. Agent $2$ DL aims to minimize latency while optimizing energy consumption. 
To coordinate these objectives, a unified global target is introduced as shared information for both agents. The rewards are defined as follows:

\begin{enumerate}
    \item UL reward $R_{u,m}^t$ -- Uploading efficiency metric where unbalanced load incurs a penalty of $-1$, balanced load yields a reward of $+1$, and that is unbalanced but within optimal capacity is assigned $0$. Therefore, the possible reward values are: $R_{u,m}^t=\{-1,0,1\}$.
    \item {DL reward $R_{d,m}^t$}- Represents energy penalty based on the agent’s energy expenditure. It is calculated as the average ratio of allocated power to maximum DL transmission power across all XUDs:  \( R_{d,m}^t = -\sum_{k=1}^{\mathcal{K}} \frac{\mathbf{p'}_{m}^t, \mathbf{p'}_{min}}{\mathbf{p'}_{max}, \mathbf{p'}_{min}}  \times 0.5\)
    \item GL reward $R_g^t$ -- Defined as the sum of the delay penalty $-1$ and the normalized total system utility $U_m^t$ achieved at time $t$:
    \( R_g^t = -1 + \sum_{m \in \mathcal{M}} U_m^t.\)
\end{enumerate}
Based on this formulation, we show that the problem constitutes an AMDP.

\begin{theorem}
Given the global state $X^t=(X_m^t)_{m\in\mathcal{M}}$ with
$X_m^t=(X_{u,m}^t,X_{d,m}^t)\in\mathcal{X}$, the operator policy
$\pi^t=\{\pi_u^t,\pi_d^t\}$ determines $(\mathbf{o}^t,\mathbf{p'}^t)$, and the instantaneous cost $C_\pi^t$ in \eqref{operator_cost_updated_final}. Then, $\mathcal{P}_O$ is an AMDP.
\end{theorem}

\textit{Proof:}
The operator's control policy is
$\pi^t=\{\pi_u^t,\pi_d^t\}$, where $\pi_u^t$ determines the OFMO vector
$\mathbf{o}^t$ and $\pi_d^t$ determines the downlink POAL vector $\mathbf{p'}^t$.
The operator aims to minimise the long-term discounted cost in $\mathcal{P}_O$:
\( 
J(X^t,\pi)=\mathbb{E}_{\pi}\!\left[\sum_{k=0}^{\infty}\gamma^{k}\,C_{\pi}^{t+k}\right],
\quad \pi^{*}=\arg\min_{\pi} J(X^t,\pi).
\)
Reward signals $(R_{u,m}^t, R_{d,m}^t, R_g^t)$ are for AMRL training, depend on the current state and actions, and do not affect state transitions. In each time slot, the uplink OFMO decision occurs first, producing an intermediate state that constrains the downlink POAL decision. Thus, the process evolves sequentially as
\(
X^{t} \xrightarrow{\ \pi_u^{t}\ } X^{t+\frac{1}{2}}
\xrightarrow{\ \pi_d^{t}\ } X^{t+1}.
\)
Accordingly, the transition kernel factorises as
\(
\Pr\!\left(X^{t+1}\mid X^{t},\pi_u^{t},\pi_d^{t}\right)
=
\sum_{x}
\Pr\!\left(X^{t+1}\mid X^{t+\frac12}=x,\pi_d^{t}\right)
\Pr\!\left(X^{t+\frac12}=x\mid X^{t},\pi_u^{t}\right).
\)
The stage rewards are generated as
$R_{u,m}^t=\mathcal{R}_{u}(X^t,\pi_u^t)$,
$R_{d,m}^t=\mathcal{R}_{d}(X^{t+\frac12},\pi_d^t)$, and
$R_g^t=\mathcal{R}_{g}(X^t,\pi_u^t,\pi_d^t)$.
Conditioned on the current state and sequential actions, the next state is independent of the past, so the process is an AMDP.

\subsection{Asynchronous Multiagent Reinforcement Learning (AMRL)}
\label{sub:AMRL}
\textbf{Novelty of our AMRL algorithm.} Our AMRL framework integrates a discrete-action UL actor, a continuous-action DL actor, and multiple hybrid-critic variants designed to handle the asymmetric decision-making inherent in INC-assisted offloading. Unlike AAHC ~\cite{yu2023asynchronous}, in which the UL agent directly selects the offloading mode, our UL actor predicts \emph{preference scores} from which a feasible OFMO is obtained via a knapsack solver. This approach aligns with our Stackelberg formulation, where XUDs determine their task split while the operator infers the optimal feasible offloading configuration. Furthermore, unlike AAHC, which employs a single multi-head critic without local critics, our AMRL framework investigates three critic architectures (AHMRL, MASC, AC) to examine how the critic design affects operator learning in in–network–edge collaboration scenarios

\subsubsection{Hybrid Critic}
\label{sec:critic_arch}

In the proposed AMRL framework, the operator must jointly optimize UL user association and DL power allocation while accounting for different reward semantics at each stage. To evaluate the impact of critic design on learning stability and convergence, we consider three value-function architectures:

\paragraph*{i) Asynchronous Hybrid Multiagent Reinforcement Learning (AHMRL)}
In contrast to AAHC~\cite{yu2023asynchronous}, AHMRL augments the multi-head hybrid critic 
with \emph{local UL and DL value networks}.  
This design enables each actor to obtain a stage-specific baseline while still benefiting from a 
shared global value estimate that captures UL--DL coupling.
The hybrid critic decomposes the value function into three components
$V_{\phi}^{u} = V_{\phi}(X_{u})$,
$V_{\phi}^{d} = V_{\phi}(X_{d})$ and
$V_{\phi}^{g} = V_{\phi}(\{X_{u};X_{d}\})$. The corresponding losses are
\begin{align}
L^{u}(\phi) &= (V_{\phi}(X_{u}^{t}) - A_{u}^{t} - V_{\phi'}(X_{u}^{t}))^{2},\nonumber\\
L^{d}(\phi) &= (V_{\phi}(X_{d}^{t}) - A_{d}^{t} - V_{\phi'}(X_{d}^{t}))^{2},\nonumber \\
L^{g}(\phi) &= (V_{\phi}(\{X_{u}^{t};X_{d}^{t}\}) - A_{g}^{t} - V_{\phi'}(\{X_{u}^{t};X_{d}^{t}\}))^{2}.
\end{align}
The hybrid critic is updated via
\begin{equation}
L(\phi) = w_{u}L^{u}(\phi) + w_{d}L^{d}(\phi) + w_{g}L^{g}(\phi),
\end{equation}
where $w_u,w_d  \text{ and }w_g$ are weighting coefficients.

\textit{ii) Multi-actor shared-critic (MASC)}. In this architecture, the UL and DL agents share a single critic \cite{xiao2023collaborative} \cite{yang2023decentralized} with no local critic for either agent. In this case, the value function is express as $ V_{\phi} ^{g} = V_{\phi} (\{X_{u}; X_{d}\})$.
The loss function is defined as
\begin{align} L^{g}(\phi) &= (V_{\phi} (\{X_{u}^{t};X_{d}^{t}\})-A_{g}^{t}{-} V_{\phi '}(\{X_{u}^{t};X_{d}^{t}\}))^{2}, \end{align}

The shared critic is updated according to this loss: 
\begin{equation}L(\phi) = L^{g}(\phi) \end{equation}

\textit{iii) Actor-critic (AC)}.
This architecture represents the prevailing approach in which agents act independently, with no shared critic. Each agent relies on its local critic to learn the optimal policy. The value functions are defined as  $V_{\phi} ^{u} = V_{\phi} (X_{u})$ and $V_{\phi} ^{d} = V_{\phi} (X_{d})$, respectively.

The losses for the three heads are given by
\begin{align} L^{u}(\phi) = (V_{\phi} (X_{u}^{t})-A_{u}^{t}{-} V_{\phi'}(X_{u}^{t}))^{2}, \nonumber\\ L^{d}(\phi) = (V_{\phi} (X_{d}^{t})-A_{d}^{t}{-} V_{\phi '}(X_{d}^{t}))^{2}. \end{align}

\subsubsection{Learning the Optimal Operator Policy}

In each episode, the system first encodes the UL and DL states into compact latent 
representations using a task-order independence layer (TIL), ensuring permutation-invariant 
processing of user-task arrivals.  
The UL actor observes the encoded state $X_{u}^{t}$ and outputs a preference vector 
$\pi_{u}^{t}$, receiving reward feedback $R_{u}^{t}$. The DL actor then observes $X_{d}^{t}$ and selects the continuous POAL $\pi_{d}^{t}$, 
receiving both $R_{d}^{t}$ and a global reward $R_{g}^{t}$.  
GAE is used to compute the advantages associated with UL, DL, and global returns, which are then used to update both actors and the critic according to the chosen architecture. This process continues until the end of the episode. While each agent is designed to achieve its own role-specific and global objectives, neither agent can observe the other’s role-specific objectives. Agent $1$ is rewarded for optimal user association in UL transmission, whereas Agent $2$ is rewarded for DL power selection.

\subsubsection{Discrete Action via Knapsack Optimization}
The first agent in our AMRL framework is an actor network whose output $ \mathbf{O} = [O_1, O_2, \ldots, O_M] $
represents OFMO \emph{preference scores}, i.e., continuous utility values indicating the benefit of INC-assisted mode for each user. These actor-generated utilities are used to determine a capacity-feasible offloading configuration.

Since INC nodes have limited processing capacity, the predicted preferences cannot be applied directly. Therefore, we formulate the OFMO selection at time~$(t+1)$ as a 0–1 knapsack problem:

\begin{align}
\widetilde{\mathcal{P}_O}:~\max_{\{b_m^{(t+1)}\}} ~~&\sum_{m=1}^{M} b_m^{(t+1)} O_m \\
\text{s.t.}\quad 
& \sum_{m=1}^{M} b_m^{(t+1)} \le C_{\max}, \tag{26a} \label{eq:constraint_22a} \\
& b_m^{(t+1)} \in \{0,1\} \tag{26b} \label{eq:constraint_22b}
\end{align}

\noindent where \(b_m^{(t+1)}\) denotes the binary OFMO decision and \(C_{\max}\) represents the INC capacity.

To solve this problem, we define a matrix \(\Xi(m,c)\) that represents the maximum utility archievable by the first \(m\) users under capacity \(c\):

\begin{equation}
\Xi(m,c)=\max\big\{\Xi(m-1,c),\; \Xi(m-1,c-1)+O_m \big\}.
\end{equation}

This recursion formulation, implemented in Algorithm~\ref{alg:ofmo_solver}, yields the optimal binary OFMO vector while ensuring that the INC capacity constraints are satisfied.

\subsubsection{Continuous Action via Sampling}

For the DL agent, the POAL decision is modeled as a continuous action.  
The continuous PPO actor outputs the parameters of a Gaussian distribution: a normalized mean
$\mu_\theta(X_d^t)$ (computed via a sigmoid layer) and a trainable log-standard deviation. An action is then
sampled as $ a_d \sim \mathcal{N}(\mu_\theta, \sigma_\theta^2),
$
clipped to the range $[0,1]$, and finally mapped to the valid DL power interval $\mathbf{p'}_m^t
= \mathbf{p'}_{\min} 
+ a_d \left(\mathbf{p'}_{\max}-\mathbf{p'}_{\min}\right).$
  This sampling mechanism ensures smooth exploration of the continuous POAL space while
adhering to practical DL transmission-power constraints \cite{chua2024play,yu2023asynchronous}.

\subsubsection{Asynchronous Actor Network Update}

The AMRL framework adopts a PPO-style policy update, where each actor is trained using a configuration-dependent advantage function. In the general case, the UL and DL policy gradients are given by

\begin{align*}
\Delta\theta_{1} &=
\mathbb{E}_{(s_{u}^{t},\pi_{u}^{t})\sim\pi_{\theta_{1}'}} 
\left[
\nabla_{\theta_{1}} f^{t}(\theta_{1},A_{u}^{t})
\right], \nonumber\\
\Delta\theta_{2} &=
\mathbb{E}_{(s_{d}^{t},\pi_{d}^{t})\sim\pi_{\theta_{2}'}} 
\left[
\nabla_{\theta_{2}} f^{t}(\theta_{2},A_{d}^{t})
\right], \tag{21}
\end{align*}

where the effective advantages $A_{u}^{t}$ and $A_{d}^{t}$ architecture of the critic network:

\begin{equation}
(A_{u}^{t},A_{d}^{t}) =
\begin{cases}
(\,A_{u}^{t} + A_{g}^{t},\; A_{d}^{t} + A_{g}^{t}\,), & \text{AHMRL}, \\[4pt]
(\,A_{g}^{t},\; A_{g}^{t}\,), & \text{MASC}, \\[4pt]
(\,A_{u}^{t},\; A_{d}^{t}\,), & \text{AC}. 
\end{cases}
\tag{23}
\end{equation}

The advantages are computed via generalized advantage estimation (GAE).  
For the UL stage:

\begin{align*}
A_{u}^{t} &= 
\delta_{u}^{t} + (\gamma\lambda)\delta_{u}^{t+1} 
+ \cdots 
+(\gamma\lambda)^{\bar{T}-1}\delta_{u}^{t+\bar{T}-1}, \tag{22}\\
\delta_{u}^{t} &= 
R_{u}^{t} + 
\gamma V_{\phi'}(X_{u}^{t+1}) - V_{\phi'}(X_{u}^{t}). \tag{23}
\end{align*}

Analogous expressions hold for $A_{d}^{t}$ and $A_{g}^{t}$, where the value
function $V(\cdot)$ is provided by either the local critic, the global critic,
or both, depending on the selected hybrid-critic configuration.

\begin{algorithm}[t]
\caption{MultiINC OFMO via 0--1 Knapsack}
\label{alg:ofmo_solver}
\begin{algorithmic}[1]

\Statex \textbf{Input:} Preference scores $\mathbf{O}=[O_1,\dots,O_M]$, capacity $C_{\max}$
\Statex \textbf{Output:} OFMO decision $\mathbf{b}^{t+1}$

\State Initialize $\Xi(0,c)=0$, $\forall c\in[0,C_{\max}]$

\For{$m=1$ to $M$}
    \For{$c=0$ to $C_{\max}$}
        \State $\Xi(m,c)=\max\{\Xi(m-1,c),\, \Xi(m-1,c-1)+O_m\}$
        \State $\Xi_r(m,c)=\arg\max_{a\in\{0,1\}}\{\Xi(m-1,c-a)+aO_m\}$
    \EndFor
\EndFor

\Statex \textbf{Backtracking:}
\For{$m=M$ downto $1$}
    \State $b_m^{t+1}=\Xi_r(m,C_{\max})$
    \State $C_{\max}\leftarrow C_{\max}-b_m^{t+1}$
\EndFor

\State \Return $\mathbf{b}^{t+1}$

\end{algorithmic}
\end{algorithm}

\begin{table}[ht]
\centering
\caption{Simulation Parameters}
\label{tab:sim_params}
\begin{tabular}{ll}
\hline
\textbf{Parameter} & \textbf{Value} \\
\hline
MEC computing capacity $f^{em}$ & 30 GHz \\
INC computing capacity $f^{cn}_k$ & [1, 9] GHz \\
INC association capacity & 5 users per node \\
Task input size $I^t_m$ & [1, 5] \\
Task computation load $C_m$ & [1, 5] \\
Task latency bound $T_m^{\max}$ & [5, 15] ms \\
DT discrepancy (INC/MEC) & 0.3 \\
Replay memory size & 10000 \\
Batch size & 32 \\
Discount factor $\gamma$ & 0.9 \\
Zipf parameter $\delta$ & 0.7 \\
No-request probability $R$ & 0.1 \\
Neighborhood size $N$ & 3 \\
\hline
\end{tabular}
\end{table}

\begin{algorithm}[t]
\caption{Nash-AMRL for Joint Operator Learning and XUD Response}
\label{alg:nash_ahmrl}
\begin{algorithmic}[1]
\State Initialize operator policies $(\pi_u,\pi_d)$ and XUD strategies via Algorithm~\ref{alg:epg_user}
\For{each episode}
    \For{each time slot $t$}
        \State Observe state $X^t$
        \State Update UL policy $\pi_u^t$ and infer OFMO $\mathbf{o}^t$
        \State Update DL policy $\pi_d^t$ and obtain POAL $\mathbf{p'}^t$
        \State Broadcast $(\mathbf{o}^t,\mathbf{p'}^t)$ to all XUDs
        \State Run Algorithm~\ref{alg:epg_user} to update XUD association and task splitting
        \State Collect rewards $(R_u^t,R_d^t,R_g^t)$ and next state $X^{t+1}$
        \State Update actor(s) and critic using PPO and GAE
    \EndFor
\EndFor
\State \Return $(\mathbf{s}^\star,\pi_u^\star,\pi_d^\star)$
\end{algorithmic}
\end{algorithm}

\section{Numerical Results}
\label{NumericalResults}
This section evaluates the performance of the proposed Digital Twin-assisted Nash-AMRL framework from both the user-oriented and operator-oriented perspectives. We benchmark the proposed MARL-based methods against heuristic policies using multiple performance metrics. Unless otherwise stated, the primary simulation parameters are summarized in Table~\ref{tab:sim_params}.

\subsection{Baseline}

This section evaluates the proposed DT-assisted Nash-AMRL framework against both learning-based variants and heuristic baselines. 

\subsubsection{Proposed Method and its variants}
The whole design integrates (i) UL/DL encoders, (ii) the hybrid critic with
UL–DL–global value heads, (iii) the knapsack-based OFMO module, and (iv) continuous
DL power control. This configuration serves as the primary algorithm under evaluation. To examine the impact of different critic structures, we implement three AMRL variants (AHMRL, MASC, AC) that support the proposed OFMO mechanism and a continuous DL action space, as discussed in \ref{sub:AMRL} and illustrated in Fig~\ref{fig:AMRL_variant}.

\subsubsection{Heuristic Baselines}
We also compare against non-learning heuristic strategies to benchmark the
performance of the proposed MARL framework.

\begin{itemize}
    \item \textit{Game-Random (GM-RN):} A randomized heuristic lower bound in which the OFMO and POAL are assigned randomly  \cite{cui2021reinforcement,pham2022partial}.

    \item \textit{Equal Policy:} Fixed 50\% collaborative offloading with uniform
    DL power allocation.

    \item \textit{Proportional Policy:} Offloading probability and DL power are
    assigned proportionally to INC resource availability and channel gain.
\end{itemize}

\begin{figure}[t]
  \centering
  \begin{subfigure}{0.52\columnwidth}
    \centering
    \includegraphics[width=\linewidth]{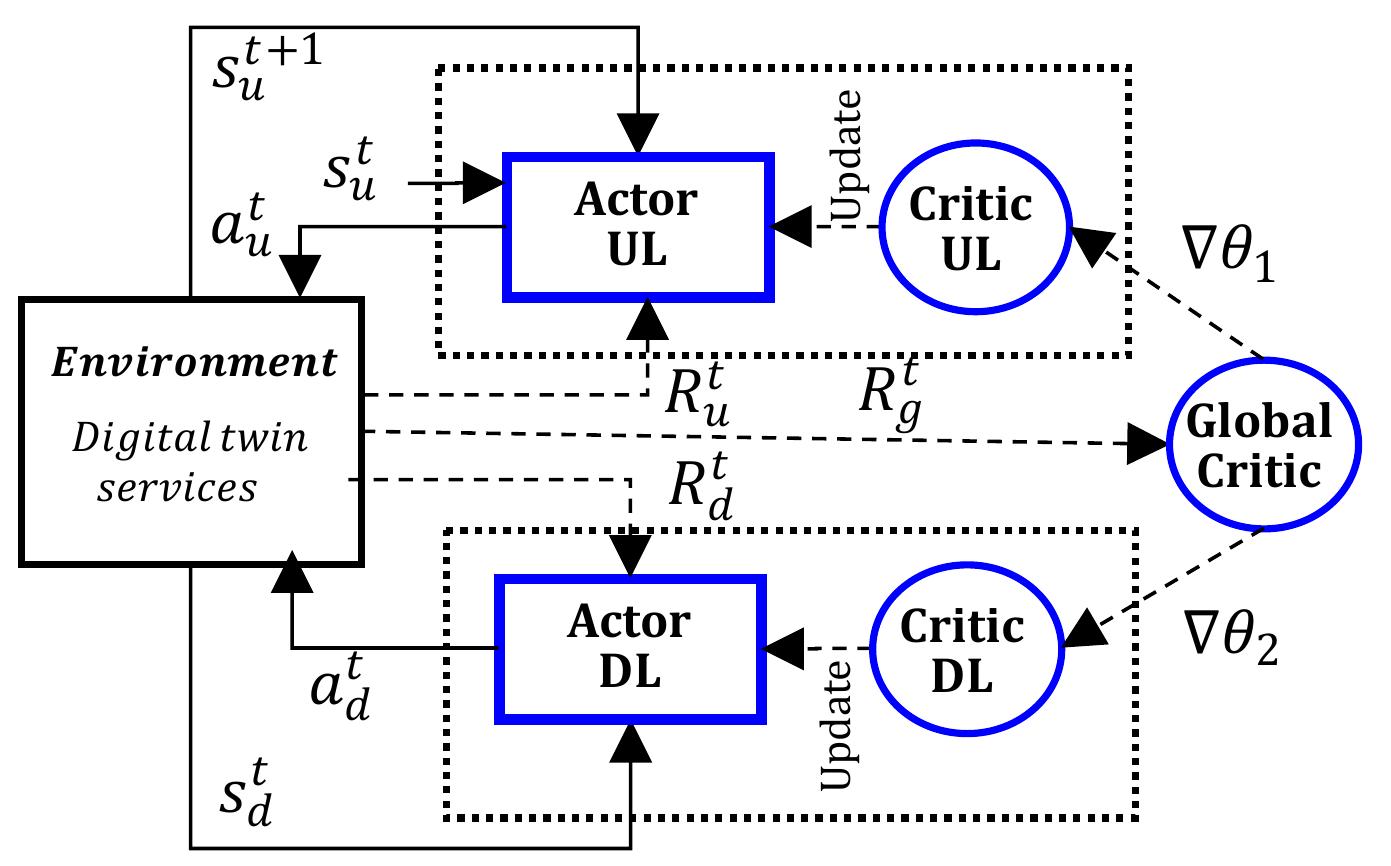}
    \caption{}
    \label{fig:AMRL_arch_AHMRL}
  \end{subfigure}
  \hfill
  \begin{subfigure}{0.42\columnwidth}
    \centering
    \includegraphics[width=\linewidth]{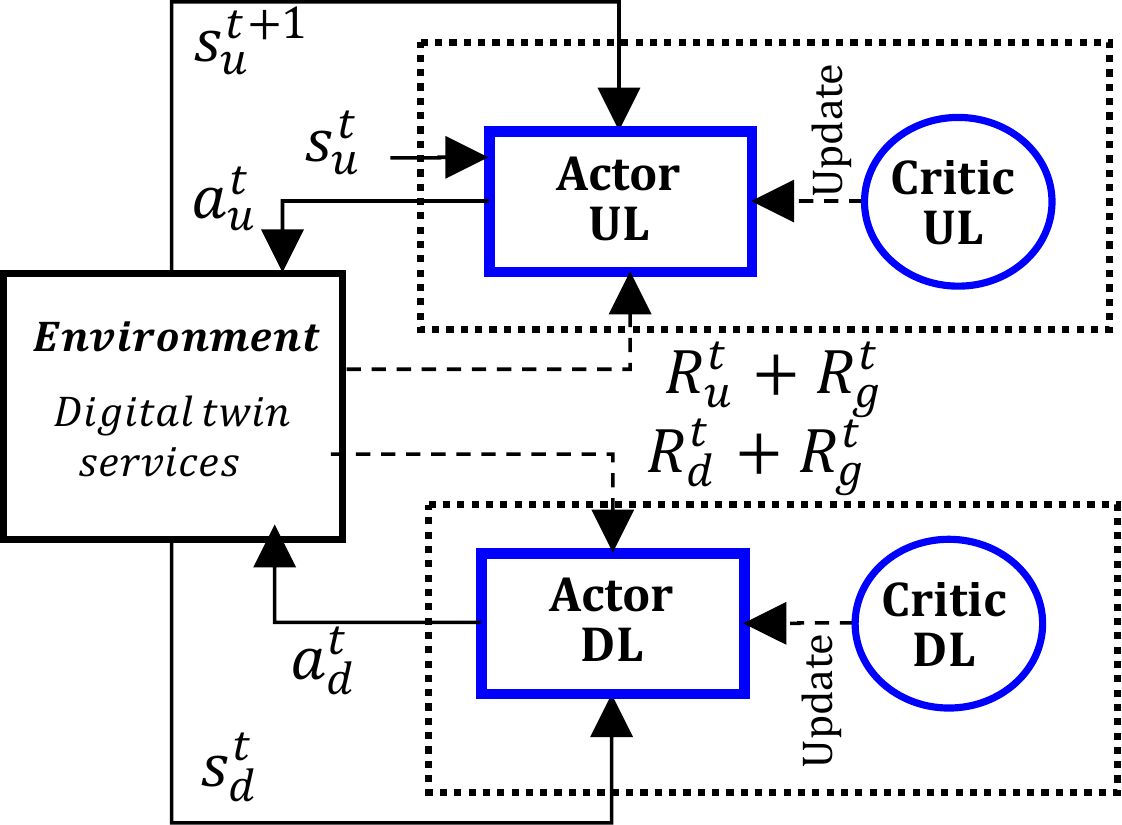}
    \caption{}
    \label{fig:AMRL_arch_MASC}
  \end{subfigure}

  \vspace{4pt} 

  \begin{subfigure}{0.45\columnwidth}
    \centering
    \includegraphics[width=\linewidth]{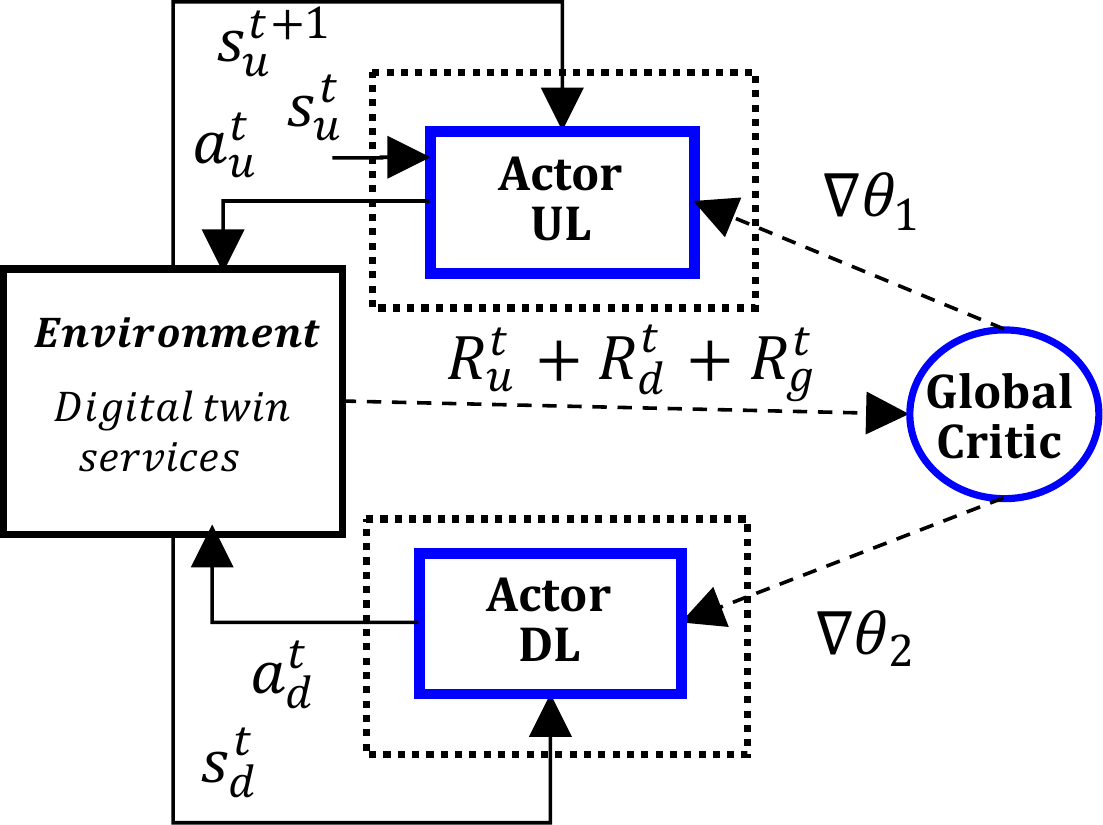}
    \caption{}
    \label{fig:AMRL_arch_AC}
  \end{subfigure}

  \caption{AMRL architectures: (a) Asynchronous Hybrid Multiagent Reinforcement Learning (AHMRL), (b) Actor--critic (AC), and (c) Multi-actor shared-critic (MASC).}
  \label{fig:AMRL_variant}
\end{figure}

\subsection{Metrics (KPIs)}
We evaluate the Nash-AMRL framework using key performance indicators (KPIs) covering user, operator, and system-level metrics.

First, we measure the (i) \textbf{utility} of XUDs, which indicates convergence of the multi-user offloading game across hybrid critic settings. We also report (ii) \textbf{maximum uplink rate} among XUDs as a measure of UL agent channel-selection efficiency, since uplink contention affects downlink rendering performance. 
Because the operator’s policy is learned via reinforcement learning, we further assess (iii) \textbf{uplink reward}, (iv) \textbf{downlink reward}, and (v) \textbf{global reward}, corresponding to the hybrid-critic components. These metrics quantify the performance of UL and DL agents in user association and power allocation.
To assess XR task performance, we measure (vi) \textbf{end-to-end latency} and (vii) \textbf{energy} consumption for digital-twinning, both of which are critical for efficient UL–DL coordination and sustainable system operation.

Furthermore, to benchmark proposed MARL architectures against heuristics baselines, we adopt a system-level (viii) \textbf{performance gain} metric: Let $C(\mathcal{A})$  denote the weighted latency–energy cost under architecture $\mathcal{A}$. The performance gain is defined as
$   \mathrm{PG}(\mathcal{A}) = \frac{C(\Pi^{\text{ref}})}{C(\mathcal{A})}
$, where $\Pi^{\text{ref}}$ denotes the heuristic reference policy. The cost function combines latency and energy consumption as $C = w_{\ell}\frac{T_{\mathrm{non}}}{T_{\mathrm{e2e}}}
      + w_{e}\frac{E_{\mathrm{non}}}{E}.$
A value $\mathrm{PG}(\mathcal{A}) > 1$ indicates that architecture $\mathcal{A}$ outperforms the heuristic baseline. This metric is used in the CDF plots of the combined cost gain.

\subsection{Experimental Settings}

This section presents the numerical results and analysis obtained from our simulations, which evaluate the performance of the proposed method. We consider INC-E networks comprising six XUDs randomly distributed within $200 \, \text{m} \times 200 \, \text{m}$ area, four INC nodes, and one edge server (ES). The large-scale fading from the $m$-th XUD to the $k$-th access point (AP) is modeled as $g_m = 10^{\left(\frac{PL(d_{mk})}{10}\right)}$, where the path loss is given by $PL(d_{mk}) = -35.3 - 37.6 \log_{10}(d_{mk})$ \cite{van2022urllc}. The noise spectral density is set to $-174 \, \text{dBm/Hz}$ \cite{nasir2020resource}, the system bandwidth is $10 \, \text{MB}$, and the URLLC decoding error probability is $\epsilon = 10^{-9}$. 
Following \cite{yu2023asynchronous}, the data augmentation function $f_m(\cdot)$ is implemented as a proportional function with slope $q^t_m$ (i.e., $I'^t_m = q^t_m \times I^t_m $), where  $q^t_m$ is uniformly drawn from the interval $[1,10]$. The choice of $q^t_m$ is motivated by advances in 3D reconstruction techniques, such asNeuralRecon \cite{sun2021neuralrecon} and Monster Mash \cite{dvorovzvnak2020monster}. The downlink transmit power range $[\mathbf{p'}_{min}, \mathbf{p'}_{max}]$ is set to $[0,20]$W. Additional simulation parameters are summarized in Table I.

Similar to \cite {chen2022dynamic}, \cite{qian2020reinforcement} and \cite{sun2018joint}, we assume that all users' task request transition probabilities follow the same model, parameterized by \( R \), \( \delta \), and \( N \):

\[
\Pr[\mu_k^{(t+1)} = j \mid \mu_k^{(t)} = i] =
\]
\[
\begin{cases} 
R, & \text{if } i \in \mathcal{F} \setminus \{0\},\; j = 0, \\[2pt]
(1-R) \dfrac{1}{\sum_{m=1}^{|\mathcal{F}|} \frac{1}{m^\delta}}, 
& \text{if } i = 0,\; j \in \mathcal{F}, \\[4pt]
(1-R)\dfrac{1}{N}, 
& \begin{aligned}[t]
\text{if } i \in \mathcal{F},\; j = (i+q)\bmod(|\mathcal{F}|+1),\\
q \in \{1,2,\ldots,N\},
\end{aligned} \\[6pt]
0, & \text{otherwise}.
\end{cases}
\]

Here,  \( R \) denotes the probability that no task is requested, and a Zipf distribution with parameter \( \delta \) models the probability of requesting any task  \( f \in \mathcal{F} \) when no current request exists. The transition probability to a neighboring task is uniformly distributed across \( N \) tasks. Although these transition probabilities are used to generate dynamic task requests in simulation, the proposed model-free solution does not rely on this information. As in \cite{chen2022dynamic}, we further demonstrate that our framework remains effective under various transition probability settings.

\begin{table}[!t]
\centering
\caption{Normalized AUC across Reward and Rate Metrics (Scenario $K$-$F$-$U$; Task Size 10, 30; $\delta \in \{0.7, 1.3\}$)}
\label{tab:normalized_auc_summary}

\resizebox{0.5\textwidth}{!}{   
\begin{tabular}{lcccccc}
\toprule
Scenario & $\delta$ & DL\ AUC & UL\ AUC & GL\ AUC & ULrate\ AUC & Util\  AUC \\
\midrule
\multicolumn{7}{c}{\textbf{AHMRL}} \\
\cmidrule(lr){1-7}
4-30-10 & 0.7 & 0.78 & 0.66 & 0.23 & 0.00 & 0.97 \\
4-30-10 & 1.3 & 0.65 & 0.64 & 0.91 & 0.74 & 0.08 \\
4-30-30 & 0.7 & 0.91 & 0.83 & 0.52 & 0.43 & 1.00 \\
4-30-30 & 1.3 & 0.62 & 0.52 & 0.31 & 0.99 & 0.81 \\
\textbf{Average} & -- & \textbf{0.74} & \textbf{0.66} & \textbf{0.49} & \textbf{0.54} & \textbf{0.72} \\
\midrule
\multicolumn{7}{c}{\textbf{MASC}} \\
\cmidrule(lr){1-7}
4-30-10 & 0.7 & 0.99 & 0.76 & 0.20 & 0.24 & 0.88 \\
4-30-10 & 1.3 & 0.99 & 1.00 & 1.00 & 0.48 & 0.17 \\
4-30-30 & 0.7 & 0.00 & 0.00 & 0.38 & 0.45 & 0.94 \\
4-30-30 & 1.3 & 0.51 & 0.48 & 0.29 & 1.00 & 0.77 \\
\textbf{Average} & -- & \textbf{0.62} & \textbf{0.56} & \textbf{0.47} & \textbf{0.54} & \textbf{0.69} \\
\midrule
\multicolumn{7}{c}{\textbf{AC}} \\
\cmidrule(lr){1-7}
4-30-10 & 0.7 & 0.06 & 0.02 & 0.00 & 0.44 & 0.86 \\
4-30-10 & 1.3 & 0.26 & 0.23 & 0.84 & 0.96 & 0.00 \\
4-30-30 & 0.7 & 1.00 & 0.86 & 0.53 & 0.30 & 0.98 \\
4-30-30 & 1.3 & 0.84 & 0.65 & 0.40 & 0.95 & 0.85 \\
\textbf{Average} & -- & \textbf{0.54} & \textbf{0.44} & \textbf{0.44} & \textbf{0.66} & \textbf{0.67} \\
\bottomrule
\end{tabular}
}
\end{table}

\subsection{Results}
We first analyze the transition probabilities to evaluate model performance under dynamic user requests. Table~\ref{tab:normalized_auc_summary} summarizes the normalized AUC results for reward, rate, and utility metrics under varying task loads and heterogeneity levels ($\delta=0.7,1.3$). The AUC metrics reflect cumulative convergence quality and overall policy effectiveness. 

AHMRL exhibits better and balanced coordination, achieving high DL and UL AUC (up to 0.91 and 0.83) and near-optimal utility (up to 1.00) at $\delta=0.7$, while shifting toward higher UL-rate efficiency (up to 0.99) at $\delta=1.3$. MASC remains competitive, attaining peak coordination with UL AUC = 1.00 and GL AUC = 1.00 at $\delta=1.3$, but shows reduced performance under higher task load at $\delta=0.7$. AC achieves the highest UL-rate AUC (up to 0.96), indicating strong transmission efficiency, but exhibits weaker and less consistent performance in reward-based metrics, particularly under lighter workloads. Therefore, all subsequent experiments adopt $\delta=0.7$ and $F=30$ users to ensure balanced and fair evaluation.

\subsubsection{Training performance} We evaluate the proposed framework using UL reward, DL reward, GL reward, utility, uplink rate, and energy consumption. All architectures converged within a few episodes. Because the GL reward is penalty-based, MASC initially exhibits a performance drop before converging to AHMRL (see Fig.~\ref{fig:train_results}a). 
The global critic in AHMRL helps maintain higher GL rewards by facilitating more effective learning across agents. For UL reward (Fig.~\ref{fig:train_results}b), AHMRL improves steadily despite a low initial value, whereas the other architectures begin to overfit after approximately 40 episodes. DL rewards generally increase during training (Fig.~\ref{fig:train_results}c); however, all architectures show some signs of overfitting around 90 episodes. 
AHMRL benefits from the global critic through enhanced knowledge sharing, while MASC learns more slowly due to reliance on a single shared critic. AC agents, despite lacking global information exchange, can still maximize their individual rewards.
Regarding overall system utility, AHMRL outperforms the other architectures for most of the training period; however, MASC surpasses it after approximately 100 episodes (Fig.~\ref{fig:train_results}d). This observation indicates that a shared global critic plays a critical role in facilitating knowledge sharing and improving overall utility in asynchronous learning. 
The maximum UL rate per episode (Fig.~\ref{fig:train_results}e) is highest for AHMRL during most of the training, with MASC and AC achieving competitive performance. Nevertheless, AHMRL’s high UL rate results in increased system energy consumption (Fig.~\ref{fig:train_results}f), revealing a clear trade-off between UL throughput and energy usage. 
AHMRL also achieves low latency up to approximately 50 episodes (Fig.~\ref{fig:train_results}g), after which AC demonstrates superior latency performance.
Relying on the OFMO and POAL decisions generated by the agents, we analyze the convergence behavior of the multi-user offloading game. As shown in Fig.~\ref{fig:train_results}h, the game rapidly converges to a stable point, corresponding to the NE. 
Moreover, AHMRL attains the highest utility, followed by MASC and AC, respectively. These results indicate that user utility is considerably enhanced by the operator’s agent, which consistently outperforms random assignment strategies.

\begin{figure}[t]
  \centering
  \includegraphics[width=0.45\textwidth]{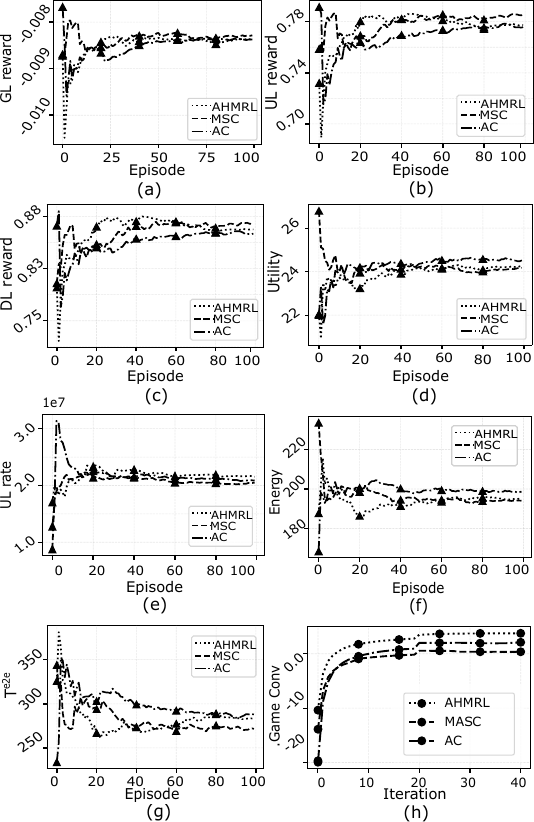}
  \caption{Training for different models and metrics (a) GL reward (b) UL reward (c) DL reward (d) Utility (e) UL rate (f) Energy (g) End-to-end Latency (h) Game utility convergence.}
  \label{fig:train_results}
\end{figure}

\subsubsection{User-Oriented Evaluation}

Under increasing user demand and task heterogeneity, we categorize subtasks as either data-intensive ($\psi{=}1,2$) or compute-intensive ($\psi{=}3,4$). For data-intensive tasks, input sizes $I^t_m$ are randomly selected from $[50, 150]$ megabits, while computation loads $C_m$ are drawn from $ [0.1,0.5]$ megacycles. 
For compute-intensive tasks, smaller input sizes of $[1,4]$ MB are paired with larger software volumes of $[1,2]$ Gcycles. AHMRL achieves the highest global reward for data-intensive workloads, outperforming MASC and AC by 6.24\% and 4.07\% in  $\psi= 1$ and $\psi= 2$, respectively. In compute-intensive scenarios, its performance degrades, with the global reward decreasing by 9.17\% and 5.71\% for  $\psi= 4$ relative to MASC and AC, respectively. 
AHMRL also demonstrates superior uplink scalability, achieving a 1.90\% higher UL rate in Task 1 and a 0.73\% higher UL reward in Task 3 compared with MASC, while offering only marginal improvements over AC, except at  $\psi= 4$ (2.15\%). Energy-delay trade-offs remain limited: AHMRL’s energy variation ranges from 0.04\% to +1.33\%, and it reduces the end-to-end latency $T_m^{\mathrm{e2e}}$ by 3.73\% relative to MASC at $\psi= 3$. However, it incurs a higher delay at $\psi= 1$. 
Overall, AHMRL excels in communication-heavy workloads, whereas MASC is more effective in compute-intensive, high-density user scenarios. Although energy consumption and latency increase with the number of users for all models, AHMRL exhibits the flattest delay growth under balanced task conditions, making it well-suited for scalable deployments.

\begin{figure}[t]
  \centering
  \includegraphics[width=0.47\textwidth]{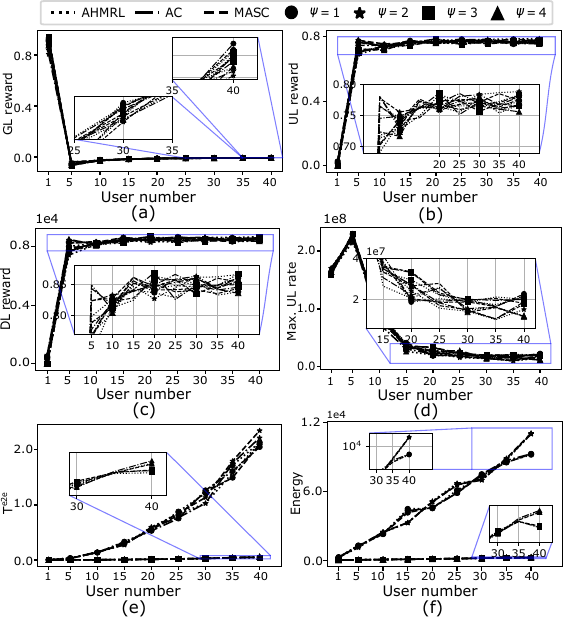}
  \caption{Influence of task type vs  number of UE}
  \label{fig:user_results}
\end{figure}

\subsubsection{Operator-Oriented Evaluation}
AHMRL consistently achieves the highest global reward under small INC $K$ configurations, outperforming MASC by +33.91\% and AC by +38.16\% at $K = 2$. However, AHMRL underperforms AC at $K = 1$ by –39.72\%, indicating suboptimal coordination when INC resources are minimal. 
For DL reward, AHMRL exhibits modest improvements over AC and MASC (up to +0.95\%), except at $K = 3$, where MASC outperforms AHMRL by +1.14\%. Regarding UL performance, AHMRL improves the UL rate over AC by 3.22\% at $K = 1$ and increases the UL reward by 0.75\% compared with MASC at $K = 2$. MASC generally maintains stronger UL performance under larger INC sizes (e.g., – 1.07\% UL reward gap at $K = 3$). 
In terms of latency, AHMRL reduces $T_m^{\mathrm{e2e}}$  by 3.04\% at $K = 4$ relative to MASC, but lags AC by 1.29\% at $K = 1$. For energy consumption, AHMRL achieves +0.66\% savings over AC at $K = 2$, but consumes more energy than MASC at large $K$.
 Overall, AHMRL performs best in coordinated scenarios with moderate INC sizes, MASC is more effective for larger $K$, and AC remains competitive when $K$ availability is minimal.

\begin{figure}[t]
  \centering
  \includegraphics[width=0.47\textwidth]{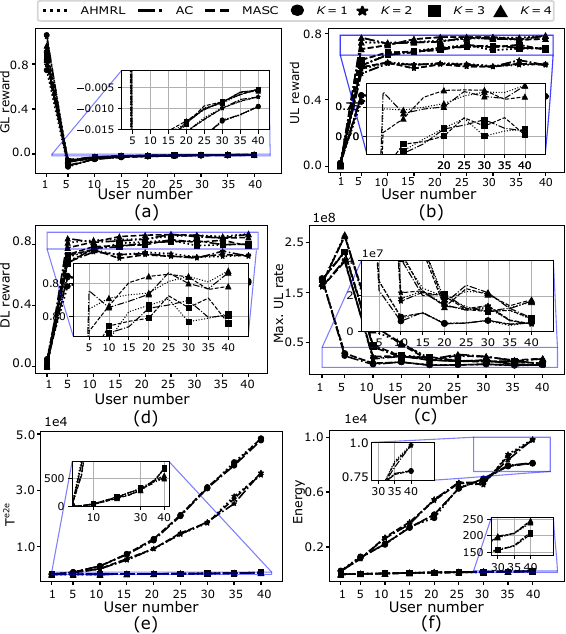}
  \caption{Influence of INC number vs  number of UE}
  \label{fig:user_results}
\end{figure}

\begin{figure}[t]
  \centering
  \includegraphics[width=0.47\textwidth]{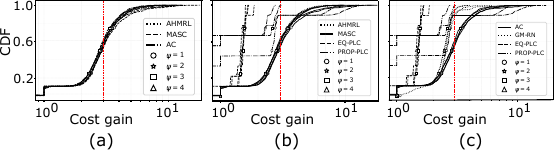}
  \caption{CDFs of user-level cost gains under different learning-based and heuristic policies}
  \label{fig:cdf_marl_user}
\end{figure}
\subsubsection{Robustness of Nash-MARL over Heuristic Baselines}
We evaluate robustness and fairness by analyzing the CDFs of user-Level and operator-level cost gains under different coordination settings.
As shown in Fig.~\ref{fig:cdf_marl_user}(a), AHMRL exhibits the steepest user CDF, with more than 85\% of users achieving at least a 2.0 cost gain, compared to 70\% for GM-RN and fewer than 60\% for AC. This indicates a more consistent user experience under AHMRL. Compared with policy-based heuristics (Fig.~\ref{fig:cdf_marl_user}(b) and ~\ref{fig:cdf_marl_user}(c)), over 90\% of users achieve cost gains of 1.5 or more with AHMRL. In contrast, fewer than 50\% do so with EQ-PLC and about 30\% with PROP-PLC, demonstrating AHMRL’s superior consistency and robustness.
On the operator side (Fig.~\ref{fig:results_CDF_operator}), AHMRL and GM-RN achieve cost gains above 2.0 in more than 80\% of scenarios, while heuristic baselines rarely reach this level.
Collectively, AHMRL provides higher and more uniform gains at both user and operator levels, validating its suitability for fairness- and reliability-critical metaverse orchestration. These results further confirm that MARL-based architectures, particularly AHMRL, offer superior performance, fairness, and robustness for equitable orchestration in dynamic metaverse environments.

\begin{figure}[t]
  \centering
  \includegraphics[width=0.47\textwidth]{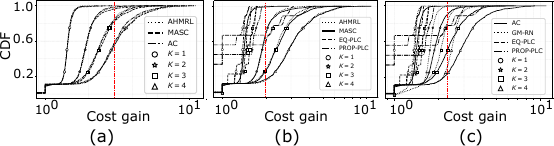}
  \caption{CDFs of operator-level cost gains under different learning-based and heuristic policies}
  \label{fig:results_CDF_operator}
\end{figure}

\subsection{Discussion}
Although AHMRL adopts a more decomposed architecture, MASC achieves comparable performance. This indicates that the global reward signal is already sufficiently informative, thereby reducing the marginal benefit of additional critic heads. Moreover, while AHMRL introduces stage-specific UL and DL value heads to enable finer-grained credit assignment, empirical results show that MASC remains highly competitive. We attribute this behavior to the characteristics of the INC-assisted MEC environment, in which the global reward signal largely dominates the optimization process. Consequently, the additional decomposition provided by AHMRL does not always yield considerable performance gains; in some scenarios, the simpler single-head MASC critic offers greater learning stability.

\textbf{Computational Complexity.} We analyze computational complexity for both training and execution. In collaboration, XUDs participate in an interference-aware channel selection to converge to a NE. Let $T_{\mathrm{NE}}$ denote the number of iterations required for convergence to a NE, and $T_{\mathrm{LR}}$ the number of Lagrangian refinement iterations for optimizing
$(\lambda_{mk},\beta_m)$. Since each of the $M$ users evaluates at most $K$
channel/INC options per iteration, the overall user-side computational cost is $\mathcal{C}_{\mathrm{XUD}} =
O\!\left(T_{\mathrm{NE}} \times M \times K \times T_{\mathrm{LR}}\right)$, matching the standard complexity of potential-game-based offloading schemes \cite{pham2022partial}.

At each time slot, the operator executes three lightweight online steps. First, the UL actor performs a forward inference to generate $M$ preference
scores, with complexity $O(d^{\mathrm{UL}} U^{\mathrm{UL}}_{1})$. Second, these scores are passed to a knapsack-based OFMO feasibility solver,
whose dynamic programming recursion requires evaluates $M$ users over the INC capacity dimension $C_{\max}$, yielding complexity $O(2 M C_{\max})$.
Finally, the DL actor performs a forward inference to produce complex continuous POAL actions
$O(d^{\mathrm{DL}} U^{\mathrm{DL}}_{1})$.
Hence, the total online execution cost per slot is
$\mathcal{C}_{\mathrm{online}}
=
O\!\left(
d^{\mathrm{UL}} U^{\mathrm{UL}}_{1}
+
d^{\mathrm{DL}} U^{\mathrm{DL}}_{1}
+
2 M C_{\max}
\right)$, which is suitable for real-time operation.

Offline training dominates computational cost. Each gradient update includes two actor updates and one critic update. For a mini-batch size $B$, the per-update complexity is
\begin{align*}
O\Bigg(
B \sum_{i \in \{\mathrm{UL},\mathrm{DL},C\}}
\Big(
d^{i} U^{i}_{1}
+ \sum_{\ell=1}^{L_{i}-1} U^{i}_{\ell} U^{i}_{\ell+1}
\Big)
\Bigg)
\end{align*}
where $C$ is the final term that corresponds to the critic network. Alternative critics (MASC or AC) adjust this term. Training complexity scales linearly with $N_{\mathrm{train}}$.
Thus, the overall complexity of per-epoch computational cost is
\( 
O\!\left(T_{\mathrm{NE}} M K T_{\mathrm{LR}}\right)
+
O\!\left(2 M C_{\max}\right)
+
O\!\left(
B(d^{\mathrm{UL}} U^{\mathrm{UL}}_{1}
+ d^{\mathrm{DL}} U^{\mathrm{DL}}_{1}
+ d^{C} U^{C}_{1})
\right),
\)
where architecture-dependent layer-wise terms are absorbed into the big-$O$ notation. In contrast, online execution involves only forward inference and one knapsack computation, resulting in lightweight real-time computational costs.

\section{Related work}
\label{RelatedWork}

Our work relates to DT-assisted offloading \cite{wu2021digital, uhlemann2017digital} and collaborative MARL \cite{hou2022device, gu2022collaborative}. Still, it introduces INC-E, which poses unique challenges: dynamic, computation intensive user requests and strict latency requirements. Unlike traditional MEC approaches that use centralized or static methods, INC-E requires distributed, adaptive decision-making in highly dynamic environments.

Most DT-assisted offloading models consider MEC as the final computation node \cite{li2022digital, masaracchia2022digital, guo2023intelligent}. While DRL-based methods\cite{lu2021adaptive, liu2021digital, zhou2022energy} enhance offloading performance, they do not address INC-E's iterative in-network processing, where tasks are partially computed at multiple nodes before final aggregation at MEC. Consequently, INC-E requires multi-stage, decentralized decision-making as data traverse the network. In contrast, \cite{gao2026deterministic} and \cite{sidoretti2025dida} focus on centralized scheduling and data-plane inference, respectively, but do not address hierarchical user–operator interactions or joint UL/DL optimization for asymmetric XR workloads.

MARL approaches for edge-cloud systems focus on distributed resource allocation and offloading \cite{hou2022device, gu2022collaborative}. Many adopt centralized training with decentralized execution (CTDE), e.g., MAPPO \cite{xiao2023collaborative}, but they lack the asynchronous learning mechanisms necessary for INC-E.

Hybrid MARL architectures incorporate hybrid reward mechanisms \cite{van2017hybrid,yu2023asynchronous,wang2018handover},yet they do not accomodate INC-E's asymmetric, multi-layered structure, in which tasks are processed across multiple nodes before MEC aggregation. We draw from HRA and its multi-agent extensions, AHHC \cite{yu2023asynchronous} and MALS \cite{chua2024play},but these remain actor-centric and are incompatible with our Stackelberg framework, in which XUDs make offloading decisions. At the same time, the operator provides support through uplink and downlink agents.

To our knowledge, this is the first study on collaborative INC-E for dynamic resource allocation in metaverse offloading using a DT-assisted framework. Our Nash-AMRL approach captures asymmetric interactions and INC-assisted coupling, enabling stable learning across operator agents and reflecting user-operator interdependence. Comparative analysis demonstrates improvements in stability, convergence, and performance trade-offs.

\section{Conclusion}
\label{Conclusion}
We study the joint optimization of user association, task offloading, and energy allocation in multiuser DT-assisted INC-E using a game-theoretic approach. The operator manages wireless and computation resources, while XUDs autonomously offload tasks to maximize utility. We address asynchronous UL association and DL energy allocation via a Stackelberg game, developing adaptive AMRL algorithms for the operator and a decentralized algorithm for XUDs to achieve Nash equilibrium. Experimental results show that our scheme is computationally efficient, improving overall utility, UL rate, and energy consumption. AHMRL achieves the best balance among energy efficiency, latency, and utility, while demonstrating scalability and robustness to task variations. MASC provides a strong alternative for latency-sensitive scenarios. All MARL models considerably outperform heuristic approaches, supporting hybrid actor-critic methods and centralized training for future in-network-edge orchestration. Future research may explore diverse user-node configurations and variations in task sizes.

\bibliographystyle{IEEEtran}
\bibliography{main}
\vspace{-3\baselineskip}

\begin{IEEEbiography}[{\includegraphics[width=1in,height=1.25in,clip,keepaspectratio]{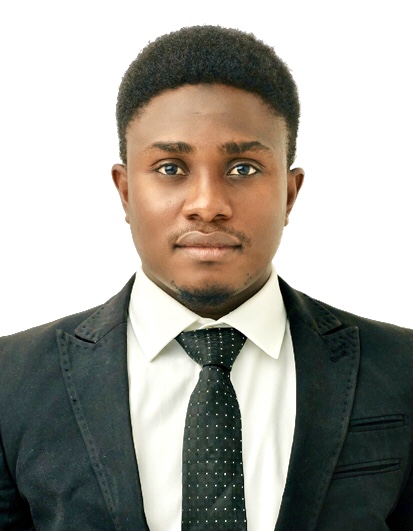}}]{Ibrahim Aliyu} (Member, IEEE) received a PhD in computer science and engineering from Chonnam National University in Gwangju, South Korea, in 2022. He also holds BEng (2014) and MEng (2018) degrees in computer engineering from the Federal University of Technology in Minna, Nigeria. He is currently a researcher with the Hyper Intelligence Media Network Platform Lab in the Department of Intelligent Electronics and Computer System Engineering at Chonnam National University. His research focuses on distributed computing for massive Metaverse deployment. His other research interests include federated learning, data privacy, network security, and artificial intelligence for autonomous networks.
Use 
\end{IEEEbiography}

\vspace{-2\baselineskip}
\begin{IEEEbiography}[{\includegraphics[width=1in,height=1.25in,clip,keepaspectratio]{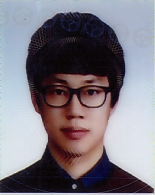}}]
    {Seungmin Oh}
 received a bachelor’s degree in the Department of Digital Contents from the Korea Nazarene University in 2019. He received his M.S. degree in the Department of ICT Convergence System Engineering from the Chonnam National University in 2021. He has been pursuing his Ph.D. in the Department of ICT Convergence System Engineering at Chonnam National University since 2021. His current research interests include deep learning and machine learning, and computer vision.
\end{IEEEbiography}

\vspace{-2\baselineskip}
\begin{IEEEbiography}[{\includegraphics[width=1in,height=1.25in,clip,keepaspectratio]{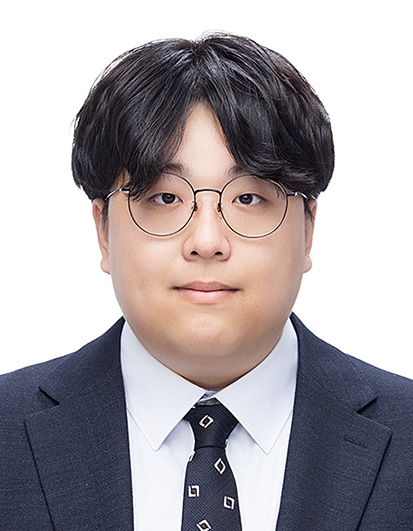}}]
    {Sangwon Oh}
 Sangwon Oh is a Ph.D. student at Chonnam National University, specializing in Generative Models, Networking, Metaverse, and AI technologies. His research primarily focuses on Diffusion-based time-series data generation and augmentation, as well as developing AI frameworks that bridge real and virtual networks. He is dedicated to advancing data analysis techniques, particularly in exploring the structural complexities of network metrics to enhance the synergy between artificial intelligence and next-generation communication systems..
\end{IEEEbiography}
\vspace{-2\baselineskip}
\begin{IEEEbiography}[{\includegraphics[width=1in,height=1.25in,clip,keepaspectratio]{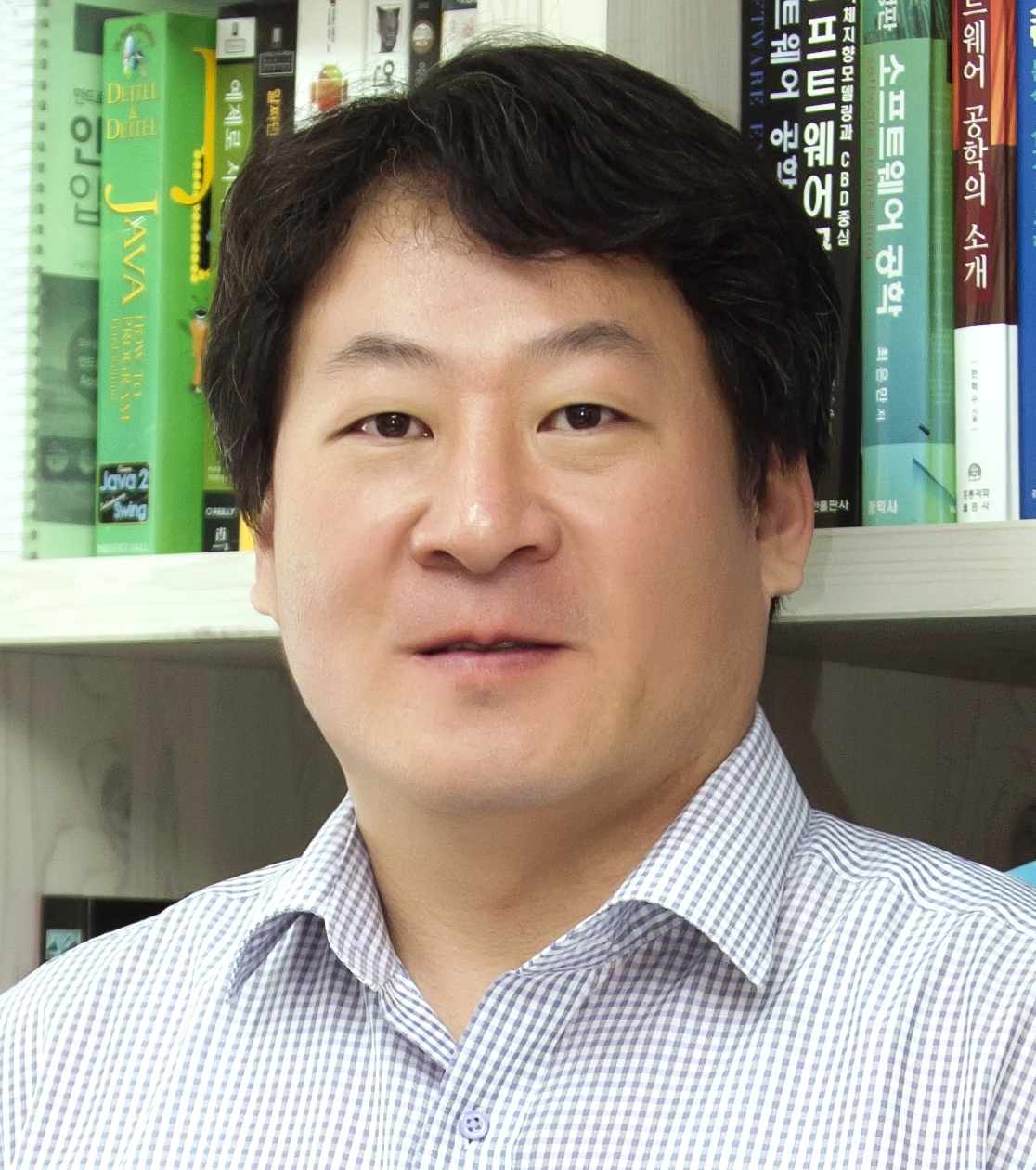}}]{Jinsul Kim}(Member, IEEE) received a BS degree in computer science from the University of Utah at Salt Lake City in Utah, USA, in 1998 and MS (2004) and PhD (2008) degrees in digital media engineering from the Korea Advanced Institute of Science and Technology (KAIST) in Daejeon, South Korea. Previously, he worked as a researcher at the Broadcasting/Telecommunications Convergence Research Division, Electronics and Telecommunications Research Institute (ETRI) in Daejeon, Korea, from 2004 to 2009 and was a professor at Korea Nazarene University in Cheonan, Korea, from 2009 to 2011. He is a professor at Chonnam National University, Gwangju, Korea. He is a member of the Korean national delegation for ITU-T SG13 international standardization. He has participated in various national research projects and domestic and international standardization activities. He is a co-research director of the Artificial Intelligence Innovation Hub Research and Development Project hosted by Korea University and is the director of the G5-AICT research center.
\end{IEEEbiography}

\end{document}